\documentclass[letterpaper,11pt]{article}
\usepackage{amsmath,amsthm,amsfonts,dsfont}
\usepackage{amssymb,latexsym,graphicx}
\usepackage{mathtools}
\usepackage{hyperref}
\usepackage{xcolor}
\hypersetup{
	colorlinks,
	linkcolor={red!75!black},
	citecolor={blue!75!black},
	urlcolor={blue!75!black}
}

\newtheorem{thm}{Theorem}[section]
\newtheorem{theorem}[thm]{Theorem}
\newtheorem{claim}[thm]{Claim}
\newtheorem{definition}[thm]{Definition}
\newtheorem{lemma}[thm]{Lemma}
\newtheorem{cor}[thm]{Corollary}

\newtheorem{fact}[thm]{Fact}
\newtheorem{prop}[thm]{Proposition}

\newcommand{\R}{\ensuremath{\mathbb{R}}}
\newcommand{\Z}{\ensuremath{\mathbb{Z}}}

\newcommand{\lat}{\mathcal{L}}
\newcommand{\cL}{\mathcal{L}}

\newcommand{\eps}{\varepsilon} 
\renewcommand{\epsilon}{\varepsilon}
\newcommand{\poly}{\mathrm{poly}}
\newcommand{\polylog}{\mathrm{polylog}}

\newcommand{\rank}{\mathrm{rank}}
\newcommand{\vol}{\mathrm{vol}}

\DeclareMathOperator*{\expect}{\mathbb{E}}
\DeclareMathOperator{\dist}{dist}
\DeclareMathOperator{\spn}{span}

\newcommand{\basis}{\ensuremath{\mathbf{B}}}

\DeclarePairedDelimiter\floor{\lfloor}{\rfloor}
\DeclarePairedDelimiter\ceil{\lceil}{\rceil}

\mathtoolsset{centercolon}

\usepackage[utf8]{inputenc}
\usepackage{mathbbold}
\usepackage{fullpage}
\usepackage[all]{xy}
\usepackage{amsmath}
\usepackage{amsthm}
\usepackage{amsfonts}
\usepackage{appendix}
\usepackage{pifont} 
\usepackage{amssymb}
\usepackage{mathrsfs}
\usepackage[T1]{fontenc}
\usepackage{ae}
\usepackage{boxedminipage}
\usepackage{bibspacing}
\usepackage{tweaklist}
\renewcommand{\enumhook}{\setlength{\topsep}{0pt}%
  \setlength{\itemsep}{0pt}}
  \usepackage{tweaklist}
\renewcommand{\itemhook}{\setlength{\topsep}{0pt}%
  \setlength{\itemsep}{0pt}}
  \usepackage{tweaklist}
\renewcommand{\descripthook}{\setlength{\topsep}{0pt}%
  \setlength{\itemsep}{0pt}}
  
\usepackage{color}
\usepackage{graphicx}
\usepackage{amssymb,amsmath,amsfonts}
\usepackage{epsfig,longtable}
\usepackage[tight,hang]{subfigure}
\usepackage{algorithm}
\usepackage{algorithmic}

\newcommand{\ignore}[1]{}
\renewcommand{\vec}[1]{\mathbf{#1}}

\newcommand{\mZ}{\mathbb{Z}}

\usepackage{caption}

\newcommand{\bareta}{\overline{\eta}}

\usepackage{multirow}

\usepackage[hyperpageref]{backref}
\title{A $2^{n/2}$-Time Algorithm for $\sqrt{n}$-SVP and $\sqrt{n}$-Hermite SVP,\\
and an Improved Time-Approximation Tradeoff for (H)SVP}
\author{Divesh Aggarwal \thanks{This work was partially supported in part by the Singapore National Research Foundation under NRF RF Award No. NRF-NRFF2013-13, the Ministry of Education, Singapore under grants MOE2012-T3-1-009, and MOE2019-T2-1-145.}\\National University of Singapore \\
\texttt{dcsdiva@nus.edu.sg} 
\and Zeyong Li\\National University of Singapore \\ \texttt{li.zeyong@u.nus.edu} \and Noah Stephens-Davidowitz\thanks{Some of this work was done at MIT, supported by an NSF-BSF grant number 1718161 and
NSF CAREER Award number 1350619, at the Centre for Quantum Technologies at the National University of Singapore, and at the Simons Institute in Berkeley.}\\
Cornell University\\
\texttt{noahsd@gmail.com}}
\date{}
\begin{document}

\maketitle

\begin{abstract}
We show a $2^{n/2+o(n)}$-time algorithm that finds a (non-zero) vector in a lattice $\lat \subset \R^n$ with norm at most $\widetilde{O}(\sqrt{n})\cdot \min\{\lambda_1(\lat), \det(\lat)^{1/n}\}$, where $\lambda_1(\lat)$ is the length of a shortest non-zero lattice vector and $\det(\lat)$ is the lattice determinant. Minkowski showed that $\lambda_1(\lat) \leq \sqrt{n} \det(\lat)^{1/n}$ and that there exist lattices with $\lambda_1(\lat) \geq \Omega(\sqrt{n}) \cdot \det(\lat)^{1/n}$, so that our algorithm finds vectors that are as short as possible relative to the determinant (up to a polylogarithmic factor).

The main technical contribution behind this result is new analysis of (a simpler variant of) a $2^{n/2 + o(n)}$-time algorithm from~\cite{ADRSSolvingShortest15}, which was only previously known to solve less useful problems. To achieve this, we rely crucially on the ``reverse Minkowski theorem'' (conjectured by Dadush~\cite{DRStrongReverse16} and proven by~\cite{Regev:2017:RMT:3055399.3055434}), which can be thought of as a partial converse to the fact that $\lambda_1(\lat) \leq \sqrt{n} \det(\lat)^{1/n}$.
    
Previously, the fastest known algorithm for finding such a vector was the $2^{.802n + o(n)}$-time algorithm due to~\cite{LWXZShortestLattice11}, which actually found a non-zero lattice vector with length $O(1) \cdot \lambda_1(\lat)$. Though we do not show how to find lattice vectors with this length in time $2^{n/2+o(n)}$, we do show that our algorithm suffices for the most important application of such algorithms: basis reduction. In particular, we show a modified version of Gama and Nguyen's slide-reduction algorithm~\cite{GNFindingShort08}, which can be combined with the algorithm above to improve the time-length tradeoff for shortest-vector algorithms in nearly all regimes---including the regimes relevant to cryptography.
\end{abstract}

\thispagestyle{empty}
\newpage
\tableofcontents
\pagenumbering{roman}
\newpage
\pagenumbering{arabic}

\section{Introduction}
    A lattice $\lat \subset \R^n$ is the set of integer linear combinations  
	\[
	\lat := \lat(\basis) = \{z_1 \vec{b}_1 + \cdots + z_n \vec{b}_n \ : \ z_i \in \Z \}
	\]
	of linearly independent basis vectors $\basis = (\vec{b}_1,\ldots, \vec{b}_n) \in \R^{n \times n}$. 
	We define the length of a shortest non-zero vector in the lattice as $\lambda_1(\lat) := \min_{\vec{x}\in \lat_{\neq \vec0}} \|\vec{x}\|$. (Throughout this paper, $\|\cdot\|$ is the Euclidean norm.)
	
	The Shortest Vector Problem (SVP) is the computational search problem whose input is a (basis for a) lattice $\lat \subseteq \R^n$, and the goal is to output a shortest non-zero vector $\vec{y} \in \lat$ with $\|\vec{y}\| = \lambda_1(\lat)$.  For $\delta \geq 1$, the $\delta$-approximate variant of SVP ($\delta$-SVP) is the problem of finding a non-zero  vector $\vec{y} \in \lat$ of length at most $\delta \cdot \lambda_1(\lat)$ given a basis of $\lat$. 
	
	$\delta$-SVP and its many relatives have found innumerable applications over the past forty years. More recently, many cryptographic constructions have been discovered whose security is based on the (worst-case) hardness of $\delta$-SVP or closely related lattice problems. See~\cite{PeiDecadeLattice16} for a survey. Such lattice-based cryptographic constructions are likely to be used in practice on massive scales (e.g., as part of the TLS protocol) in the not-too-distant future~\cite{NISPostQuantumCryptography18}.
	
	For most applications, it suffices to solve $\delta$-SVP for superconstant approximation factors. E.g., cryptanalysis typically requires $\delta = \poly(n)$. However, our best algorithms for $\delta$-SVP work via (non-trivial) reductions to $\delta'$-SVP for much smaller $\delta'$ \emph{over lattices with smaller rank}, typically $\delta' = 1$ or $\delta' = O(1)$. E.g., one can reduce $n^c$-SVP with rank $n$ to $O(1)$-SVP with rank $n/(c+1)$ for constant $c \geq 1$~\cite{GNFindingShort08,ALNSSlideReduction19}. Such reductions are called \emph{basis reduction algorithms}~\cite{LLLFactoringPolynomials82,SchHierarchyPolynomial87,SELatticeBasis94}.
	
	Therefore, even if one is only interested in $\delta$-approximate SVP for large approximation factors, algorithms for $O(1)$-SVP are still relevant. (We make little distinction between exact SVP and $O(1)$-SVP in the introduction.)
	
    \subsection{Sieving for constant-factor-approximate SVP}	
	
     There is thus a very long line of work~\cite{KanImprovedAlgorithms83,AKSSieveAlgorithm01,NVSieveAlgorithms08,PSSolvingShortest09,MVDeterministicSingle13,LWXZShortestLattice11,WLWFindingShortest15,ADRSSolvingShortest15,ASJustTake18,AUV19} on this problem.

	The fastest known algorithms for $O(1)$-SVP run in time $2^{O(n)}$. With one exception (\cite{MVDeterministicSingle13}), all known algorithms with this running time are \emph{sieving algorithms}. These algorithms work by sampling $2^{O(n)}$ not-too-long lattice vectors $\vec{y}_1,\ldots, \vec{y}_M \in \lat$ from some nice distribution over the input lattice $\lat$, and performing some kind of sieving procedure to obtain $2^{O(n)}$ shorter vectors $\vec{x}_1,\ldots, \vec{x}_m
	\in \lat$. They then perform the sieving procedure again on the $\vec{x}_k$, and repeat this process many times. 
	
	The most natural sieving procedure was originally studied by Ajtai, Kumar, and Sivakumar~\cite{AKSSieveAlgorithm01}. This procedure simply takes $\vec{x}_k := \vec{y}_i - \vec{y}_j \in \lat$, where $i,j$ are chosen so that $\|\vec{y}_i - \vec{y}_j\| \leq (1-\eps)\min_\ell \|\vec{y}_\ell\|$. In particular, the resulting sieving algorithm clearly finds progressively shorter lattice vectors at each step. So, it is trivial to show that this algorithm will eventually find a short lattice vector. Unfortunately (and maddeningly), it seems very difficult to say nearly anything else about the distribution of the vectors when this very simple sieving technique is used, and in particular, while we know that the vectors must be short, we do not know how to show that they are \emph{non-zero}.~\cite{AKSSieveAlgorithm01} used clever tricks to modify the above procedure into one for which they could prove correctness, and the current state-of-the-art is a $2^{0.802n}$-time algorithm for $\gamma$-SVP for a sufficiently large constant $\gamma > 1$~\cite{LWXZShortestLattice11,WLWFindingShortest15,AUV19}.
	
	In this work, we are more interested in the ``sieving by averages'' technique, introduced in~\cite{ADRSSolvingShortest15} to obtain a $2^{n+o(n)}$-time algorithm for exact SVP. This sieving procedure takes $\vec{x}_k := (\vec{y}_i + \vec{y}_j)/2$ to be the average of two lattice vectors. Of course, $\lat$ is not closed under taking averages, so one must choose $i,j$ so that that $(\vec{y}_i + \vec{y}_j)/2 \in \lat$. This happens if and only if $\vec{y}_i, \vec{y}_j$ lie in the same \emph{coset} of $2\lat$, $\vec{y}_i = \vec{y}_j \bmod 2\lat$. Equivalently, the coordinates of $\vec{y}_i$ and $\vec{y}_j$ in the input basis should have the same parities. So, these algorithms pair vectors according to their cosets (and ignore all other information about the vectors) and take their averages $\vec{x}_k = (\vec{y}_i + \vec{y}_j)/2$. 
	
	The analysis of these algorithms centers around the \emph{discrete Gaussian distribution} $D_{\lat, s}$ over a lattice, given by
	\[
	    \Pr_{\vec{X} \sim D_{\lat, s}}[\vec{X} = \vec{y}] \propto e^{-\pi \|\vec{y}\|^2/s^2}
	\]
	for a parameter $s > 0$ and any $\vec{y} \in \lat$.
	When the starting vectors come from this distribution, we are able to say quite a bit about the distribution of the vectors at each step. (Intuitively, this is because this algorithm only uses algebraic properties of the vectors---their cosets---and entirely ignores the geometry.)  In particular,~\cite{ADRSSolvingShortest15} used a careful rejection sampling procedure to guarantee that the vectors at each step are distributed exactly as $D_{\lat,s}$ for some parameter $s > 0$. Specifically, in each step the parameter lowers by a factor of $\sqrt{2}$, which is exactly what one would expect, taking intuition from the continuous Gaussian.
	More closely related to this work is~\cite{ASJustTake18}, which showed that this rejection sampling procedure is actually unnecessary. 
	
	In addition to the above,~\cite{ADRSSolvingShortest15,NSDthesis} also present a $2^{n/2+o(n)}$-time algorithm that samples from $D_{\lat, s}$ as long as the parameter $s > 0$ is not too small. In particular, we need $s$ to be ``large enough that $D_{\lat,s}$ looks like a continuous Gaussian.'' This algorithm is similar to the $2^{n + o(n)}$-time algorithms in that it starts with independent discrete Gaussian vectors with some high parameter, and it gradually lowers the parameter using a rejection sampling procedure together with a procedure that takes the averages of pairs of vectors that lie in the same coset modulo some sublattice (with index $2^{n/2 + o(n)}$). But, it fails for smaller parameters because the rejection sampling procedure that it uses must throw out too many vectors in this case. (In~\cite{NSDthesis}, a different rejection sampling procedure is used that never throws away too many vectors, but it is not clear how to implement it in $2^{n/2+o(n)}$ time for small parameters $s < \sqrt{2} \eta_{1/2}(\lat)$.) It was left as an open question whether there is a suitable variant of this algorithm that works for small parameters, which would lead to an algorithm to solve SVP in $2^{n/2 + o(n)}$ time. For example, perhaps we could show that this algorithm solves SVP without doing any rejection sampling at all, as we showed for the $2^{n + o(n)}$-time algorithm in~\cite{ASJustTake18}.
	
	\subsection{Hermite SVP}
	
	We will also be interested in a variant of SVP called Hermite SVP (HSVP). HSVP is defined in terms of the determinant $\det(\lat) := |\det(\basis)|$ of a lattice $\lat$ with basis $\basis$. (Though a lattice can have many bases, one can check that $|\det(\basis)|$ is the same for all such bases, so that this quantity is well-defined.) Minkowski's celebrated theorem says that $\lambda_1(\lat) \leq O(\sqrt{n}) \cdot \det(\lat)^{1/n}$, and Hermite's constant $\gamma_n = \Theta(n)$ is the maximal value of $\lambda_1(\lat)^2/\det(\lat)^{2/n}$. (Hermite SVP is of course named in honor of Hermite and his study of $\gamma_n$. It is often alternatively called Minkowski SVP.)
	
	For $\delta \geq 1$, it is then natural to define $\delta$-HSVP as the variant of SVP that asks for any non-zero lattice vector $\vec{x} \in \lat$ such that $\|\vec{x}\| \leq \delta \det(\lat)^{1/n}$. One typically takes $\delta \geq \sqrt{\gamma_n} \geq \Omega(\sqrt{n})$, in which case the problem is total. In particular, there is a trivial reduction from $\delta\sqrt{\gamma_n}$-HSVP to $\delta$-SVP. (There is also a non-trivial reduction from $\delta^2$-SVP to $\delta$-HSVP for $\delta \geq \sqrt{\gamma_n}$~\cite{LovAlgorithmicTheory86}.)
	
	$\delta$-HSVP is an important problem in its own right. In particular, the random lattices most often used in cryptography typically satisfy $\lambda_1(\lat) \geq \Omega(\sqrt{n}) \cdot \det(\lat)^{1/n}$, so that for these lattices $\delta$-HSVP is equivalent to $O(\delta/\sqrt{n})$-SVP. This fact is quite useful because the best known basis reduction algorithms~\cite{GNFindingShort08,MWPracticalPredictable16,ALNSSlideReduction19} yield solutions to both $\delta_S$-SVP and $\delta_H$-HSVP with, e.g.,
	\begin{equation}
	\label{eq:ALNS}
	\delta_H :=  \gamma_{k}^{\frac{n-1}{2(k-1)}} \approx k^{n/(2k)}
	\qquad
	\delta_S :=
 \gamma_{k}^{\frac{n-k}{k-1}} \approx k^{n/k-1}
	\; ,
	\end{equation}
	when given access to an oracle for (exact) SVP in dimension $k \leq n/2$. Notice that $\delta_H$ is significantly better than the approximation factor $\sqrt{\gamma_n} \delta_S  \approx \sqrt{n} k^{n/k-1}$ that one obtains from the trivial reduction to $\delta_S$-SVP. (Furthermore, the approximation factor $\delta_H$ in Eq.~\eqref{eq:ALNS} holds for any $k \leq n$.)
	
	In fact, it is easy to check that we will achieve the same value of $\delta_H$ if the reduction is instantiated with a $\sqrt{\gamma_k}$-HSVP oracle in dimension $k$, rather than an SVP oracle. More surprisingly, a careful reading of the proofs in~\cite{GNFindingShort08,ALNSSlideReduction19} shows that a $\sqrt{\gamma_k}$-HSVP oracle is ``almost sufficient'' to even solve $\delta_S$-SVP. (We make this statement a bit more precise below.)

\subsection{Our results}

Our main contribution is a simplified version of the $2^{n/2 + o(n)}$-time algorithm from~\cite{ADRSSolvingShortest15} and a novel analysis of the algorithm that gives an approximation algorithm for both SVP and HSVP.

\begin{theorem}[Informal, approximation algorithm for (H)SVP]\label{thm:2^n/2_intro}
There is a $2^{n/2 + o(n)}$-time algorithm that solves $\delta$-SVP and $\delta$-HSVP for $\delta \leq \widetilde{O}(\sqrt{n})$. 
\end{theorem}
Notice that this algorithm almost achieves the best possible approximation factor $\delta$ for HSVP since there exists a family of lattices for which $\lambda_1(\cL) \geq \Omega(\sqrt{n} \det(\cL)^{1/n})$ (i.e., $\gamma_n \geq \Omega(n)$). So, $\delta$ is optimal for HSVP up to a polylogarithmic factor.

As far as we know, this algorithm might actually solve exact or near-exact SVP, but we do not know how to prove this. However, by adapting the basis reduction algorithms of~\cite{GNFindingShort08,ALNSSlideReduction19}, we show that Theorem~\ref{thm:2^n/2_intro} is nearly as good (when combined with known results) as a $2^{k/2}$-time algorithm for exact SVP in $k$ dimensions, in the sense that we can already nearly match Eq.~\eqref{eq:ALNS} in time $2^{k/2 + o(k)}$ with this.

In slightly more detail, basis reduction procedures break the input basis vectors $\vec{b}_1,\ldots, \vec{b}_n$ into blocks $\vec{b}_{i+1},\ldots, \vec{b}_{i+k}$ of length $k$. They repeatedly call their oracle on (projections of) the lattices generated by these blocks and use the result to update the basis vectors. We observe that the procedures in~\cite{GNFindingShort08,ALNSSlideReduction19} only need to use an SVP oracle on the last block  $\vec{b}_{n-k+1},\ldots, \vec{b}_n$. For all other blocks, an HSVP oracle suffices. Since we now have a faster algorithm for HSVP than we do for SVP, we make this last block a bit smaller than the others, so that we can solve (near-exact) SVP on the last block in time $2^{k/2+ o(k)}$. 

When we instantiate this idea with the $2^{0.802n}$-time algorithm for $O(1)$-SVP from~\cite{LWXZShortestLattice11,WLWFindingShortest15,AUV19}, it yields the following result. Together with Theorem~\ref{thm:2^n/2_intro}, this yields the fastest known algorithms for $\delta$-SVP for all $\delta \gtrsim n^{1/2}$.

\begin{theorem}[Informal]\label{thm:basis_reduction_intro} 
There is a $2^{k/2 + o(k)}$-time algorithm that solves  $\delta_H^*$-HSVP with
\[
	\delta_H^* \approx k^{n/(2k)}
	\; ,
\]
for $k \leq n$
and $\delta_S^*$-SVP with
\[
    \delta_S^* \approx k^{(n/k)-0.62}
    \; ,
\]
for $k \leq n/1.63$.
\end{theorem}

Notice that Theorem~\ref{thm:basis_reduction_intro} matches Eq.~\eqref{eq:ALNS} with block size $k$ exactly for $\delta_H$, and up to a factor of $k^{0.37}$ for $\delta_S$. This small loss in approximation factor comes from the fact that our last block is slightly smaller than the other blocks.

Together, Theorems~\ref{thm:2^n/2_intro} and~\ref{thm:basis_reduction_intro} give the fastest proven 
running times 
for $n^c$-HSVP for all $c > 1/2$ and for $n^c$-SVP for all $c > 1$, as well as $c \in (1/2,0.802)$. Table \ref{Table:approximat SVP} summarizes the current state of the art.

\begin{table}[t] 
\begin{center}
	\scalebox{0.88}{
	\begin{tabular}{|l|l|l l| l l |}%
		\hline
		Problem &Approximation factor & \multicolumn{2}{|l|}{Previous Best} & \multicolumn{2}{|l|}{This work} \\ \hline
		 \multirow{5}{*}{SVP} &Exact  &$2^n$ [*] &\cite{ADRSSolvingShortest15}  &  --- &\\ %
		 &$O(1)$ &$2^{0.802n}$ [*] &\cite{WLWFindingShortest15}  &  --- &\\
			 &$n^c$ for $c\in (0.5,0.802]$  &$2^{\frac{0.401n}{c}}$ &\cite{ALNSSlideReduction19}  &  $2^{\frac{n}{2}}$[*] &\\
				 &$n^c$ for $c\in (0.802,1]$  &$2^{\frac{0.401n}{c}}$ &\cite{ALNSSlideReduction19}  &  --- &\\
				 &$n^c$ for $c > 1$ &$2^{\frac{0.802 n}{c+1}}$ &\cite{ALNSSlideReduction19}  &  $2^{\frac{n}{2c+1.24}}$ &\\
				 \hline
				\multirow{2}{*}{HSVP} &$\sqrt{n}$\rule{0pt}{2.3ex}  &$2^{0.802n}$ [*] &\cite{WLWFindingShortest15}  &  $2^{\frac{n}{2}}$ [*] &\\
					 &$n^c$ for $c \ge 1$  &$2^{\frac{0.401n}{c}}$ &\cite{ALNSSlideReduction19}  &  $2^{\frac{n}{4c}}$ &\\
	 \hline
	\end{tabular}
	}
		\caption{\label{Table:approximat SVP}Proven running times for solving (H)SVP. We mark results that do not use basis reduction with [*]. We omit $2^{o(n)}$ factors in the running time, and except in the first two rows, polylogarithmic factors in the approximation factor.
	}
\end{center}
\end{table}

\subsection{Our techniques}
\label{sec:techniques}

\subsubsection{Summing vectors over a tower of lattices}

Like the $2^{n/2 + o(n)}$-time algorithm in~\cite{ADRSSolvingShortest15}, our algorithm for $\widetilde{O}(\sqrt{n})$-(H)SVP constructs a tower of lattices  $\cL_0 \supset \cL_1 \supset \cdots \supset \cL_{\ell} = \cL$ such that for every $i\ge 1$, $2\cL_{i-1} \subset \cL_i$. The index of $\cL_i$ over $\cL_{i-1}$ is $2^{\alpha}$ for an integer $\alpha = n/2 + o(n)$, and $\ell = o(n)$. For the purpose of illustrating our ideas, we make a simplifying assumption here that $\ell \alpha$ is an integer multiple of $n$, and hence $\cL_0 = \cL/2^{\alpha \ell/n}$ is a scalar multiple of $\cL$.

And, as in~\cite{ADRSSolvingShortest15}, we start by sampling $\vec{X}_1,\ldots, \vec{X}_N \in \lat_0$ for $N = 2^{\alpha + o(n)}$ from $D_{\lat_0, s}$. This can be done efficiently using known techniques, as long as $s$ is large relative to, e.g., the length of the shortest basis of $\lat_0$~\cite{GPVTrapdoorsHard08,BLPRS13}. Since $\lat_0 = \cL/2^{\alpha \ell/n}$, the parameter $s$ can still be significantly smaller than, e.g., $\lambda_1(\lat)$. In particular, we can essentially take $s \leq \poly(n) \lambda_1(\lat)/2^{\alpha \ell/n}$.

The algorithm then takes disjoint pairs of vectors that are in the same coset of $\cL_0/\cL_1$, and adds the pairs together. Since $2\cL_0 \subset \cL_1$, for any such pair $\vec{X}_i, \vec{X}_i$, $\vec{Y}_k = \vec{X}_i + \vec{X}_j$ is in $\cL_1$.  (This adding is analogous to the averaging procedure from~\cite{ADRSSolvingShortest15,ASJustTake18} described above. In that case, $\cL_1 = 2\cL_0$, so that it is natural to divide vectors in $\cL$ by two, while here adding seems more natural.) We thus obtain approximately $N/2$ vectors in $\cL_1$ (up to the loss due to the vectors that could not be paired), and repeat this procedure many times, until finally we obtain vectors in $\lat_\ell = \lat$, each the sum of $2^\ell$ of the original $\vec{X}_i$.

To prove correctness, we need to prove that with high probability some of these vectors will be \emph{both} short and non-zero. It is actually relatively easy to show that the vectors are short---at least in expectation. To prove this, we first use the fact that the expected squared norm of the $\vec{X}_i$ is bounded by $n s^2$ (which is what one would expect from the continuous Gaussian distribution). And, the original $\vec{X}_i$ are distributed symmetrically, i.e., $\vec{X}_i$ is as likely to equal $-\vec{x}$ as it is to equal $\vec{x}$). 

Furthermore, our pairing procedure is symmetric, i.e., if we were to replace $\vec{X}_i$ with $-\vec{X}_i$, the pairing procedure would behave identically. (This is true precisely because $2\lat_{0} \subset \lat_1$---we are using the fact that $\vec{x} = -\vec{x} \bmod \lat_1$ for any $\vec{x} \in \lat_0$.) This implies that 
\[
\expect[\langle \vec{X}_i,\vec{X}_j\rangle \ | \ E_{i,j} ] = \expect[\langle \vec{X}_i,-\vec{X}_j\rangle \ | \ E_{i,j} ] = 0
\; ,
\]
where $E_{i,j}$ is the event that $\vec{X}_i$ is paired with $\vec{X}_j$.
Therefore,
\[
     \expect[\|\vec{X}_i + \vec{X}_j\|^2 \ | \ E_{i,j}] = \expect[\|\vec{X}_i\|^2 \ | \ E_{i,j}] + \expect[\|\vec{X}_j\|^2 \ | \ E_{i,j}] + 2\expect[\langle\vec{X}_i, \vec{X}_j \rangle \ | \ E_{i,j}] \approx 2\expect[\|\vec{X}_i\|^2]
     \; .
\]
The same argument works at every step of the algorithm. So, (if we ignore the subtle distinction between $\expect[\|\vec{X}_i\|^2 \ | \ E_{i,j}]$ and $\expect[\|\vec{X}_i\|^2]$), we see that our final vectors have expected squared norm 
\begin{equation}
    \label{eq:length_of_output_intro}
    2^\ell \expect[\|\vec{X}_i\|^2] \leq 2^\ell n s^2 \leq \poly(n) 2^{\ell(1-2\alpha n)} \cdot \lambda_1(\lat)^2
    \; .
\end{equation}
By taking, e.g., $\alpha = n/2 + n/\log n < n + o(n)$ and $\ell = \log^2 n$, we see that we can make this expectation small relative to $\lambda_1(\lat)$.

The difficulty, then, is ``only'' to show that the distribution of the final vectors is not heavily concentrated on zero. Of course, we can't hope for this to be true if, e.g., the expectation in Eq.~\eqref{eq:length_of_output_intro} is much smaller than $\lambda_1(\lat)^2$. And, as we will discuss below, if we choose $\alpha$ and $\ell$ so that this expectation is sufficiently large, then techniques from prior work can show that the probability of zero is low. Our challenge is therefore to bound the probability of zero for the largest choices of $\alpha$ and $\ell$ (and therefore the lowest expectation in Eq.~\eqref{eq:length_of_output_intro}) that we can manage.

\subsubsection{Gaussians over unknown sublattices}

Peikert and Micciancio (building on prior work) showed what they called a ``convolution theorem'' for discrete Gaussians. Their theorem says that the sum of discrete Gaussian vectors is statistically close to a discrete Gaussian (with parameter increased by a factor of $\sqrt{2}$), provided that the parameter $s$ is a bit larger than the \emph{smoothing parameter} $\eta(\lat)$ of the lattice $\lat$~\cite{MPHardnessSIS13}. This (extremely important) parameter $\eta(\lat)$, was introduced by Micciancio and Regev~\cite{MR07}, and has a rather technical (and elegant) definition. (See Section~\ref{sec:smoothing}.) Intuitively, $\eta(\lat)$ is minimal such that for any $s > \eta(\lat)$, $D_{\lat, s}$ ``looks like a continuous Gaussian distribution.'' E.g., for $s > \eta(\lat)$, the moments of the discrete Gaussian distribution are quite close to the moments of the continuous Gaussian distribution (with the same parameter). 

In fact,~\cite{MPHardnessSIS13} showed a convolution for lattice \emph{cosets}, not just lattices, i.e., the sum of a vector sampled from $D_{\lat + \vec{t}_1, s}$ and a vector sampled from $D_{\lat + \vec{t}_2, s}$ yields a vector distributed as $D_{\lat + \vec{t}_1 + \vec{t}_2, \sqrt{2} s}$. Since our algorithm sums vectors sampled from a discrete Gaussian over $\lat_0$, conditioned on their cosets modulo $\lat_1$, it is effectively summing discrete Gaussians over cosets of $\lat_1$. So, as long as we stay above the smoothing parameter of $\lat_1 \supset \lat$, our vectors will be statistically close to discrete Gaussians, allowing us to easily bound the probability of zero.

However,~\cite{ADRSSolvingShortest15} already showed how to use a variant of this algorithm to obtain samples from \emph{exactly} the discrete Gaussian above smoothing. And, more generally, there is a long line of work that uses samples from the discrete Gaussian above smoothing to find ``short vectors'' from a lattice, but the length of these short vectors is always proportional to $\eta(\lat)$. The problem is that in general $\eta(\lat)$ can be arbitrarily larger than $\lambda_1(\lat)$ and $\det(\lat)^{1/n}$. (To see this, consider the two-dimensional lattice generated by $(T,0), (0,1/T)$ for large $T$, which has $\eta(\lat) \approx T$, $\lambda_1(\lat) = 1/T$ and $\det(\lat) = 1$.) So, this seems useless for solving (H)SVP, instead yielding a solution to another variant of SVP called SIVP.\footnote{It is not known how to use an SIVP oracle for basis reduction, which makes it significantly less useful than SVP. \cite{MR07,MPHardnessSIS13} and other works used these ideas to reduce SIVP to the problem of breaking a certain cryptosystem, in order to argue that the cryptosystem is secure. They were therefore primarily interested in SIVP as an example of a hard lattice problem, rather than as a problem that one might actually wish to solve.}

Our solution is essentially to apply these ideas from~\cite{MPHardnessSIS13} to an \emph{unknown} sublattice $\lat' \subseteq \lat$. (Here, one should imagine a sublattice generated by fewer than $n$ vectors. Jumping ahead a bit, the reader might consider the example $\lat' = \Z \vec{v} = \{\vec0, \pm \vec{v}, \pm 2\vec{v},\ldots,\}$, the rank-one sublattice generated by $\vec{v}$, a shortest non-zero vector in the lattice.) Indeed, the discrete Gaussian over $\lat$, $D_{\lat, s}$, can be viewed as a \emph{mixture of discrete Gaussians} over cosets of $\lat'$, $D_{\lat, s} = D_{\lat' + \vec{C}, s}$, where $\vec{C} \in \lat/\lat'$ is some random variable over cosets of $\lat'$. (Put another way, one could obtain a sample from $D_{\lat, s}$ by first sampling a coset $\vec{C} \in \lat/\lat'$ from some appropriately chosen distribution and then sampling from $D_{\lat' + \vec{C}, s}$.) 

The basic observation behind our analysis is that we can now apply (a suitable variant of)~\cite{MPHardnessSIS13}'s convolution theorem in order to see that the sum of two mixtures of Gaussians over $\lat'$, $\vec{X}_1, \vec{X}_2 \sim D_{\lat' + \vec{C}, s}$, yields a new mixture of Gaussians $D_{\lat' + \vec{C}', \sqrt{2} s}$ for \emph{some} $\vec{C}'$, provided that $s$ is sufficiently large \emph{relative to $\eta(\lat')$}. 

Ignoring \emph{many} technical details, this shows that our algorithm can be used to output a distribution of the form $D_{\lat' + \vec{C}, s}$ for some random variable $\vec{C} \in \lat/\lat'$ provided that $s \gg \eta(\lat')$. Crucially, we only need to consider $\lat'$ in the analysis; the algorithm does not need to know what $\lat'$ is for this to work. Furthermore, we do not care at all about the distribution of $\vec{C}$! We already know that our algorithm samples from a distribution that is short in expectation (by the argument above), so that the only thing we need from the distribution $D_{\lat' + \vec{C}, s}$ is that it is not zero too often. Indeed, when $\vec{C}$ is not the zero coset (i.e., $\vec{C} \notin \lat'$), then $D_{\lat' + \vec{C},s}$ is never zero, and when $\vec{C}$ is zero, then we get a sample from $D_{\lat', s}$ for $s \gg \eta(\lat')$, in which case well-known techniques imply that we are unlikely to get zero.

\subsubsection{Smooth sublattices}

So, in order to prove that our algorithm finds short vectors, it remains to show that there exists some sublattice $\lat' \subseteq \lat$ with low smoothing parameter---a ``smooth sublattice.'' In more detail, our algorithm will find a non-zero vector with length less than $\sqrt{n} \cdot \eta(\lat')$ for any sublattice $\lat'$. Indeed, as one might guess, taking $\lat' = \Z \vec{v} = \{\vec0, \pm \vec{v}, \pm 2 \vec{v},\ldots,\}$ to be the lattice generated by a shortest non-zero vector $\vec{v}$, we have $\eta(\lat') = \polylog(n) \|\vec{v}\| = \polylog(n)\lambda_1(\lat)$ (where the polylogarithmic factor arises because of ``how smooth we need $\lat'$ to be''). This immediately yields our $\widetilde{O}(\sqrt{n})$-SVP algorithm.

To solve $\widetilde{O}(\sqrt{n})$-HSVP, we must argue that every lattice has a sublattice $\lat' \subseteq \lat$ with $\eta(\lat') \leq \polylog(n) \cdot \det(\lat)^{1/n}$. In fact, for very different reasons, Dadush conjectured \emph{exactly} this statement (phrased slightly differently), calling it a ``reverse Minkowski conjecture''~\cite{DRStrongReverse16}. (The reason for this name might not be clear in this context, but one can show that this is a partial converse to Minkowski's theorem.) Later, Regev and Stephens-Davidowitz proved the conjecture~\cite{Regev:2017:RMT:3055399.3055434}. Our result then follows from this rather heavy hammer.

\subsection{Open questions and directions for future work}

We leave one obvious open question: Does our algorithm (or some variant) solve $\gamma$-SVP for a better approximation factor? It is clear that our current analysis cannot hope to do better than $\delta \approx \sqrt{n}$, but we see no fundamental reason why the algorithm cannot achieve, say, $\delta = \polylog(n)$ or even $\delta = 1$! (Indeed, we have been trying to prove something like this for roughly five years.) 

We think that even a negative answer to this question would also be interesting. In particular, it is not currently clear whether our algorithm is ``fundamentally an HSVP algorithm.'' For example, if one could show that our algorithm fails to output vectors of length $\polylog(n) \cdot \lambda_1(\lat)$ for some family of input lattices $\lat$, then this would be rather surprising. Perhaps such a result would be our first hint at a true algorithmic separation between the optimal running times for the two problems. 

\section{Preliminaries}

 We write $\log$ for the base-two logarithm. We use the notation $a = 1\pm \delta$ and $a=e^{\pm \delta}$ to denote the statements $1-\delta \leq a \leq 1+\delta$ and $e^{-\delta} \leq a \leq e^{\delta}$, respectively.

\begin{definition}
	We say that a distribution $\widehat{D}$ is $\delta$-similar to another distribution $D$ if for all $\vec{x}$ in the support of $D$, we have
	\[
	\Pr_{\vec{X}\sim \widehat{D}}[\vec{X} = \vec{x}] =  e^{\pm \delta} \cdot \Pr_{\vec{X}\sim D}[\vec{X} = \vec{x}]
	\; .
	\]
\end{definition}

\subsection{Probability}

The following inequality gives a concentration result for the values of (sub-)martingales that have bounded differences. 
\begin{lemma} [{{\cite[Azuma's inequality, Chapter 7]{AS04}}}]
\label{lem:Azuma}
Let $X_0, X_1, \ldots$ be a set of random variables that form a discrete-time sub-martingale, i.e., for all $n\ge 0$, \[\mathbb{E}[X_{n+1} \: | \: X_1, \ldots, X_n] \ge X_n \;. \] If for all $n \ge 0$, $|X_{n} - X_{n-1}| \le c$, then for all integers $N$ and positive real $t$, \[\Pr[X_N - X_0 \le - t] \le \exp\left(\frac{ -t^2}{2 Nc^2}\right) \; . \]
\end{lemma}

We will need the following corollary of the above inequality.
\begin{cor}
\label{cor:Azuma}
Let $\alpha \in (0,1)$, and let $ Y_1, Y_2, Y_3, \ldots$ be random variables in $[0,1]$ such that for all $n \ge 0$
\[
\mathbb{E}[Y_{n+1}|Y_1, \ldots, Y_n] \ge \alpha \;.
\]
Then, for all positive integers $N$ and positive real $t$, 
\[
\Pr[\sum_{i=1}^N Y_i \le N \alpha - t] \le \exp\left(\frac{ -t^2}{2 N}\right) \; . \]
\end{cor}
\begin{proof}
  Let $X_0 = 0$, and  for all $i \ge 1$,
  \[
  X_i := X_{i-1} + Y_i - \alpha = \sum_{j=1}^i Y_i - i \cdot \alpha \;.
  \]

  The statement then follows immediately from Lemma~\ref{lem:Azuma}. 
    \end{proof}

\subsection{Lattices}

A lattice $\lat \subset \R^n$ is the set of integer linear combinations 
	\[
	\lat := \lat(\basis) = \{z_1 \vec{b}_1 + \cdots + z_k \vec{b}_k \ : \ z_i \in \Z \}
	\]
	of linearly independent basis vectors $\basis = (\vec{b}_1,\ldots, \vec{b}_k) \in \R^{n \times k}$. 
	We call $k$ the \emph{rank} of the lattice. Given a lattice $\lat$, the basis is not unique.
 For any lattice $\lat$, we use $\rank(\lat)$ to denote its rank. 
We use $\lambda_1(\lat)$ to denote the length of the shortest non-zero vector in $\lat$,  and more generally, for $1 \leq i \leq k$,
\[
    \lambda_i(\lat) := \min\{r \ : \ \dim \spn(\{ \vec{y} \in \lat \ : \ \|\vec{y}\| \leq r\}) \geq i\} \; .
\]

For any lattice $\lat \subset \R^n$, its dual lattice $\lat^*$ is defined to be the set of vectors in the span of $\lat$ that have integer inner products with all vectors in $\lat$. More formally:
\[\lat^* = \{\vec{x} \in \spn(\lat): \forall \vec{y} \in \lat, \langle \vec{x},\vec{y}\rangle \in \Z \}\; .\]

We often assume without loss of generality that the lattice is full rank, i.e., that $n = k$, by identifying $\spn(\lat)$ with $\R^k$. However, we do often work with sublattices $\lat' \subseteq \lat$ with $\rank(\lat') < \rank(\lat)$.

For any sublattice $\lat' \subseteq \lat$, $\lat/\lat'$ denotes the set of cosets which are translations of $\lat'$ by vectors in $\lat$. In particular, any coset can be denoted as $\lat' + \vec{c}$ for $\vec{c} \in \lat$. When there is no ambiguity, we drop the $\lat'$ and use $\vec{c}$ to denote a coset.

\subsection{The discrete Gaussian Distribution}
For any parameter $s > 0$, we define Gaussian mass function $\rho_s : \R^n \rightarrow \R$ to be:
\[\rho_s(\vec{x}) = \exp\Big(-\frac{\pi \Vert \vec{x} \Vert^2}{s^2}\Big)\; , \]
and for any discrete set $A \subset \R^n$, its Gaussian mass is defined as $\rho_s(A) = \sum_{\vec{x} \in A} \rho_s(\vec{x})$.

For a lattice $\lat \subset \R^n$, shift $\vec{t} \in \R^n$, and parameter $s >0$, we have the following convenient formula for the Gaussian mass of the lattice coset $\lat + \vec{t}$, which follows from the Poisson Summation Formula
\begin{equation}
\label{eq:PSF_gaussian}
    \rho_s(\lat + \vec{t}) = \frac{s^n}{\det(\lat)} \cdot \sum_{\vec{w} \in \lat^*} \rho_{1/s}(\vec{w}) \cos(2\pi \langle \vec{w}, \vec{t} \rangle)
    \; .
\end{equation}
In particular, for the special case $\vec{t} = \vec0$, we have $\rho_s(\lat) = s^n \rho_{1/s}(\lat^*)/\det(\lat)$. 

\begin{definition}
	For a lattice $\lat \subset \R^n$, $\vec{u} \in \R^n$, the discrete Gaussian distribution $\mathcal{D}_{\lat+\vec{u},s}$ over $\lat + \vec{u}$ with parameter $s>0$ is defined as follows. For any $\vec{x} \in \lat + \vec{u}$,
	\[\Pr_{\vec{X}\sim \mathcal{D}_{\lat+\vec{u},s}}[\vec{X} = \vec{x}] = \frac{\rho_s(\vec{x})}{\rho_s(\lat + \vec{u})} \; .\]
\end{definition}

We will need the following result about the discrete Gaussian distribution. 
\begin{lemma} [{{\cite[Lemma 2.13]{DRS14}}}]
\label{lem:DGStail}
For any lattice $\lat \subset \R^n$, $s>0$, $\vec{u} \subset \R^n$, and $t > \frac{1}{\sqrt{2\pi}}$, 
\[
\Pr_{\vec{X} \sim \mathcal{D}_{\lat+\vec{u},s}}(\|\vec{X}\| > ts\sqrt{n}) < \frac{\rho_s(\lat)}{\rho_s(\lat+\vec{u})} \left(\sqrt{2\pi e t^2}\exp(-\pi t^2)\right)^n \; .
\]
\end{lemma}

\subsection{The smoothing parameter}
\label{sec:smoothing}

\begin{definition}
	For a lattice $\lat \subset \R^n$ and $\epsilon > 0$, the smoothing parameter $\eta_\epsilon(\lat)$ is defined as the unique value that satisfies $\rho_{1/\eta_\epsilon(\lat)}(\lat^* \backslash \{\vec{0}\}) = \epsilon$.
\end{definition}

We will often use the basic fact that $\eta_\eps(\alpha\lat) = \alpha \eta_\eps(\lat)$ for any $\alpha > 0$ and $\eta_\eps(\lat') \geq \eta_\eps(\lat)$ for any \emph{full-rank} sublattice $\lat' \subseteq \lat$.

\begin{claim}[{{\cite[Lemma  3.3]{MR07}}}]
\label{clm:Z_smooth}
    For any $\eps \in (0,1/2)$, we have
    \[
        \eta_\eps(\Z) \leq \sqrt{\log(1/\eps)} \; .
    \]
\end{claim}

We will need the following simple results, which follows immediately from Eq.~\eqref{eq:PSF_gaussian}. 
\begin{lemma}[{{\cite[Claim 3.8]{Reg09}}}]
\label{lem:smooth}
For any lattice $\lat$, $s \geq \eta_\epsilon(\lat)$, and any vectors $\vec{c}_1, \vec{c}_2$, we have that
\[
\frac{1-\eps}{1+\eps} \le \frac{\rho_s(\lat + \vec{c}_1)}{\rho_s(\lat + \vec{c}_2)} \le \frac{1+\eps}{1-\eps} \;.
\]
Thus, for $\eps < 1/3$, 
\[
e^{-3\eps} \le \frac{\rho_s(\lat + \vec{c}_1)}{\rho_s(\lat + \vec{c}_2)} \le e^{3\eps} \;.
\]
\end{lemma}

We prove the following  statement. 
\begin{theorem}
\label{thm:eta_lambda_n}
    For any lattice $\lat \subset \R^n$ with rank $k \ge 20$,
    \[
        \eta_{1/2}(\lat) \geq \lambda_k(\lat)/\sqrt{k}
        \; .
    \]
\end{theorem}
\begin{proof}
    If $\lat$ is not a full-rank lattice, then we can project to a subspace given by the span of $\lat$. So, without loss of generality, we assume that $\lat$ is a full-rank lattice, i.e., $k = n$.
    
    Suppose $\lambda_n(\lat) > \sqrt{n} \eta_{1/2} (\lat)$. Then there exists a vector $\vec{u} \in \R^n$ such that $\dist(\vec{u}, \lat) > \frac{1}{2}\sqrt{n} \eta_{1/2} (\lat)$. Then, using Lemma~\ref{lem:DGStail} with $t=1/2$, $s = \eta_{1/2}(\lat)$, we have 
    \begin{align*}
1 &=\Pr_{\vec{X} \sim \mathcal{D}_{\lat+\vec{u},\eta_{1/2}(\lat)}}\left[\|\vec{X}\| > st\sqrt{n} \right] \\
&< \frac{\rho_s(\lat)}{\rho_s(\lat+\vec{u})} \left(\sqrt{2\pi e t^2}\exp(-\pi t^2)\right)^n \\
& \le \frac{1+1/2}{1-1/2} (\sqrt{\pi e/2} \cdot e^{-\pi/4})^n & \text{using Lemma~\ref{lem:smooth}}\\
& \le 3 \cdot (0.943)^n \\
& < 1 &\text{since } \; k=n \ge 20\;, 
\end{align*}
 which is a contradiction.
\end{proof}

\begin{claim}
\label{clm:smooth_mass_growth}
    For any lattice $\lat \subset \R^n$ and any parameters $s \geq s' \geq \eta_{1/2}(\lat)$,
    \[
        \frac{\rho_s(\lat)}{\rho_{s'}(\lat)} \geq \frac{2s}{3s'}
        \; .
    \]
\end{claim}
\begin{proof}
    By the Poisson Summation Formula (Eq.~\eqref{eq:PSF_gaussian}), we have
    \[
        \rho_{s}(\lat) = s^n \cdot \frac{\rho_{1/s}(\lat^*)}{\det(\lat)} \geq s^n/\det(\lat)
        \; ,
    \]
    and similarly,
    \[
        \rho_{s'}(\lat) = (s')^n \cdot \frac{\rho_{1/s'}(\lat^*)}{\det(\lat)} \leq 3(s')^n/(2\det(\lat))
        \; ,
    \]
    since $\rho_{1/s'}(\lat^*) \leq 3/2$ for $s' \geq \eta_{1/2}(\lat)$.
    Combining the two inequalities gives $\rho_s(\lat) \geq 2(s/s')^n/3 \geq 2(s/s')/3$, as needed.
\end{proof}

\begin{claim}
\label{clm:DGS_variance}
    For any lattice $\lat \subset \R^n$ and any $s > 0$,
    \[
        \expect_{\vec{X} \sim D_{\lat, s}}[\|\vec{X}\|^2] \leq \frac{n s^2}{2\pi}
        \; .
    \]
\end{claim}

\begin{lemma}\label{lem:smoothing_parameter}
	For $s\geq \eta_\epsilon(\lat)$, and any real factor $k \geq 1$, $ks \geq \eta_{\epsilon^{k^2}}(\lat)$. 
\end{lemma}
\begin{proof}
	\begin{align*}
		\sum_{\vec{w} \in \lat^* \backslash \{0\}} \rho_{1/(ks)}(\vec{w}) &= \sum_{\vec{w} \in \lat^* \backslash \{0\}} e^{-\pi \|\vec{w}\| k^2s^2}\\ 
		&= \sum_{\vec{w} \in \lat^* \backslash \{0\}} \rho_{1/s}(\vec{w})^{k^2} \\
		&\leq \Big(\sum_{\vec{w} \in \lat^* \backslash \{0\}} \rho_{1/s}(\vec{w}) \Big)^{k^2} \\ 
		&\leq \epsilon^{k^2} \; .
	\end{align*}
\end{proof}

\begin{cor}\label{cor:smoothing_parameter_scaled} 
	For any lattice $\lat \subset \R^n$ and $\epsilon \in (0,1/2)$, $\eta_{\epsilon}(\lat) \leq \sqrt{\log (1/\epsilon)} \cdot\eta_{1/2}(\lat)$.
\end{cor}
\begin{proof}
	Let $k = \sqrt{\log (1/\epsilon)}$ and thus $(\frac{1}{2})^{k^2} = \epsilon$. By Lemma \ref{lem:smoothing_parameter}, $k\eta_{1/2}(\lat)\geq \eta_{\epsilon}(\lat)$.
\end{proof}
We will need the following useful lemma concerning the convolution of two discrete Gaussian distributions. See~\cite{GMPWImprovedDiscrete20} for a very general result of this form (and a list of similar results). Our lemma differs from those in~\cite{GMPWImprovedDiscrete20} and elsewhere in that we are interested in a stronger notion of statistical closeness: point-wise multiplicative distance, rather than statistical distance. One can check that this stronger variant follows from the proofs in~\cite{GMPWImprovedDiscrete20}, but we give a separate proof for completeness.

\begin{lemma}
    \label{lem:convolution}
    For any lattice $\lat \subset \R^n$, $\eps \in (0,1/3)$, parameter $s \geq \sqrt{2}\eta_\eps(\lat)$, and shifts $\vec{t}_1, \vec{t}_2 \in \R^n$, let $\vec{X}_i \sim D_{\lat + \vec{t}_i,s}$ be independent random variables. Then the distribution of $\vec{X}_1 + \vec{X}_2$ is $6\eps$-similar to $D_{\lat + \vec{t}_1 + \vec{t}_2, \sqrt{2} s}$.
\end{lemma}
\begin{proof}
    Let $\vec{y} \in \lat + \vec{t}_1 + \vec{t}_2$. We have
    \begin{align*}
        \Pr[\vec{X}_1 + \vec{X}_2 = \vec{y}] 
            &= \frac{1}{\rho_s(\lat + \vec{t}_1) \rho_s(\lat + \vec{t}_2)}\sum_{\vec{x} \in \lat + \vec{t}_1} \exp(-\pi (\|\vec{x}\|^2 + \|\vec{y} - \vec{x}\|^2)/s^2)\\
            &= \frac{1}{\rho_s(\lat + \vec{t}_1) \rho_s(\lat + \vec{t}_2)}\sum_{\vec{x} \in \lat + \vec{t}_1} \exp(-\pi (\|\vec{y}\|^2/2 + \| 2\vec{x} - \vec{y}\|^2/2)/s^2)\\
            &= \frac{\rho_{\sqrt{2} s}(\vec{y})}{\rho_s(\lat + \vec{t}_1) \rho_s(\lat + \vec{t}_2)} \rho_{s/\sqrt{2}}(\lat + \vec{t}_1 - \vec{y}/2)\\
            &= e^{\pm 3\eps} \rho_{\sqrt{2} s}(\vec{y}) \cdot \frac{\rho_{s/\sqrt{2}}(\lat)}{\rho_{s}(\lat + \vec{t}_1) \rho_s(\lat + \vec{t}_2)}
            \; ,
    \end{align*}
    where the last step follows from Lemma~\ref{lem:smooth}.
    By applying this for all $\vec{y}' \in \lat + \vec{t}_1 + \vec{t}_2$, we see that
    \[
        \Pr[\vec{X}_1 + \vec{X}_2 = \vec{y}] = e^{\pm 3\eps}\cdot \frac{\rho_{\sqrt{2}s}(\vec{y})}{\sum_{\vec{y}' \in \lat + \vec{t}_1 + \vec{t}_2} \chi_{\vec{y}'} \rho_{\sqrt{2}s}(\vec{y}')}
    \]
    for some $\chi_{\vec{y}'} = e^{\pm 3\eps}$. Therefore,
    \[
        \Pr[\vec{X}_1 + \vec{X}_2 = \vec{y}] = e^{\pm 6\eps} \cdot  \frac{\rho_{\sqrt{2} s}(\vec{y})}{\rho_{\sqrt{2} s}(\lat + \vec{t}_1 + \vec{t}_2)}
        \; ,
    \]
    as needed.
\end{proof}

\subsection{Lattice problems}
In this paper, we study the algorithms for the following lattice problems.
\begin{definition}[$r$-HSVP]
	For an approximation factor $r := r(n) \geq 1$, the $r$-Hermite Approximate Shortest Vector Problem ($r$-HSVP) is defined as follows: Given a basis $\vec{B}$ for a lattice $\lat \subset \R^n$, the goal is to output a vector $\vec{x} \in \lat\backslash \{\vec{0}\}$ with $\Vert \vec{x} \Vert \leq r \cdot \det(\lat)^{1/n}$. 
\end{definition}

\begin{definition}[$r$-SVP]
	For an approximation factor $r := r(n) \geq 1$, the $r$-Shortest Vector Problem ($r$-SVP) is defined as follows: Given a basis $\vec{B}$ for a lattice $\lat \subset \R^n$, the goal is to output a vector $\vec{x} \in \lat\backslash \{\vec{0}\}$ with $\Vert \vec{x} \Vert \leq r \cdot \lambda_1(\lat)$. 
\end{definition}

 It will be convenient to define a generalized version of SVP, of which HSVP and SVP are special cases.

\begin{definition}[$\eta$-GSVP]
    For a function $\eta$ mapping lattices to positive real numbers, the $\eta$-Generalized Shortest Vector Problem $\eta$-GSVP is defined as follows: Given a basis $\vec{B}$ for a lattice $\lat \subset \R^n$ and a length bound $d \geq \eta(\lat)$, the goal is to output a vector $\vec{x} \in \lat\backslash \{\vec{0}\}$ with $\Vert \vec{x} \Vert \leq d$.
\end{definition}

To recover $r$-SVP or $r'$-HSVP, we can take $\eta(\lat) = r \lambda_1(\lat)$ or $\eta(\lat) = r' \det(\lat)^{1/n}$ respectively. Below, we will set $\eta$ to be a new parameter, which in particular will satisfy $\eta(\lat) \leq \widetilde{O}(\sqrt{n}) \cdot  \min\{\lambda_1(\lat), \det(\lat)^{1/n} \}$.

\subsection{Gram-Schmidt orthogonalization}
For any given basis $\basis = (\vec{b}_1,\ldots, \vec{b}_n) \in \R^{m \times n}$, we define the sequence of projections $\pi_{i} := \pi_{\{\vec{b}_1,\ldots, \vec{b}_{i-1}\}^\perp}$ where $\pi_{W^\perp}$ refers to the orthogonal
projection onto the subspace orthogonal to $W$.
As in~\cite{GNFindingShort08, ALNSSlideReduction19}, we use $\basis_{[i,j]}$ to denote the projected block $(\pi_{i}(\vec{b}_i),\pi_{i}(\vec{b}_{i+1}),\ldots, \pi_{i}(\vec{b}_j))$.

The Gram-Schmidt orthogonalization (GSO) $\vec{B}^{\ast} := (\mathbf{b}_1^{\ast}, \ldots, \mathbf{b}_{n}^{\ast})$ of a basis $\basis$ is as follows: for all $i \in [1, n], \vec{b}_i^* := \pi_{i}(\vec{b}_i) =  \vec{b}_i - \sum_{j < i} \mu_{i,j} \vec{b}_j^*$, where $\mu_{i,j} = \langle \vec{b}_i, \vec{b}_j^* \rangle/\|\vec{b}_j^*\|^2$. 

\begin{theorem}[{{\cite[Lemma 3.1]{GPVTrapdoorsHard08}}}]
\label{thm:eta_GSO}
    For any lattice $\lat \subset \R^n$ with basis $\basis := (\vec{b}_1,\ldots, \vec{b}_n)$ and any $\eps \in (0,1/2)$,
    \[
        \eta_\eps(\lat) \leq \sqrt{\log(n/\eps)}\cdot \max_i \|\vec{b}_i^*\|
        \; .
    \]
\end{theorem}

 For $\gamma \geq 1$, a basis is \emph{$\gamma$-HKZ-reduced} if for all $i \in \{1,\ldots, n\}$, $\|\vec{b}_i^*\| \leq \gamma \cdot \lambda_1(\pi_i(\lat))$.

We say that a basis $\basis$ is \emph{size-reduced} if it satisfies the following condition: for all $i \neq j$, $|\mu_{i,j}| \le \frac{1}{2}$. A size-reduced basis $\basis$ satisfies that $\|\basis\| \le \sqrt{n}\|\basis^{\ast}\|$,  where $\|\basis\|$ is the length of the longest basis vector in $\basis$. It is known that we can efficiently transform any basis into a size-reduced basis while maintaining the lattice generated by the basis $\lat(\basis)$ as well as the GSO $\basis^{\ast}$. We call such operation {\em size reduction}.
\subsection{Some lattice algorithms}
\begin{theorem}[\cite{LLLFactoringPolynomials82}]\label{thm:LLL}
Given a basis $\vec{B} \in \mathbb{Q}^{n\times n}$, there is an algorithm that computes a vector $\vec{x} \in \lat(\vec{B})$ of length at most $ 2^{n/2} \cdot \lambda_1(\lat(\vec{B}))$ in polynomial time.
\end{theorem}

 We will prove a strictly stronger result than the theorem below in the sequel, but this weaker result will still prove useful.

\begin{theorem}[\cite{ADRSSolvingShortest15,GNFindingShort08}]
    \label{thm:BKZ}
    There is a $2^{r + o(r)}\cdot \poly(n)$-time algorithm that takes as input a (basis for a) lattice $\lat \subset \R^n$ and $2 \leq r \leq n$ and outputs a $\gamma$-HKZ-reduced basis for $\lat$, where $\gamma := r^{n/r}$.
\end{theorem}

\begin{theorem}[\cite{BLPRS13}]\label{thm:DG_initialisation}
	There is a probabilistic polynomial-time algorithm that takes as input a basis $\basis$ for an n-dimensional lattice $\lat \subset \R^n$, a parameter $s \geq \Vert \basis^*\Vert\sqrt{10 \log n}$ and outputs a vector that is distributed as $\mathcal{D}_{\lat,s}$, where $\Vert \basis^* \Vert$ is the length of the longest vector in the Gram-Schmidt orthogonalization of $\basis$.
\end{theorem}

\subsection{Lattice basis reduction}
\label{subsec:lattice reduction}

\paragraph{LLL reduction.} A basis $\vec{B} = (\mathbf{b}_1,
\ldots, \mathbf{b}_n)$ is $\eps$-{\em LLL-reduced} \cite{LLLFactoringPolynomials82} for $\eps \in [0, 1]$ if it is a size-reduced basis and for  $1 \leq i < n$, the projected  %
block $\vec{B}_{[i,i+1]}$
satisfies Lov\'{a}sz's condition:
$\|\mathbf{b}_{i}^{\ast}\|^2 \leq (1 +
\varepsilon)\|\mu_{i,i-1}\mathbf{b}_{i-1}^{\ast} + \mathbf{b}_{i}^{\ast}\|^2$ .
For $\eps \geq 1/\poly(n)$, an $\eps$-LLL-reduced basis for any given lattice can be computed efficiently.

\paragraph{SVP reduction and its extensions.}

Let $\vec{B} = (\mathbf{b}_1, \ldots, \mathbf{b}_{n})$ be a basis of a lattice $\lat$ 
and $\delta \ge 1$ be approximation factors. 

We say that $\vec{B}$ is $\delta$-{\em SVP-reduced} if 
$\|\mathbf{b}_1\|  \leq \delta \cdot \lambda_{1}(\lat)$.
Similarly, we say that $\vec{B}$ is $\delta$-{\em HSVP-reduced}
if $\|\vec{b}_1\| \leq \delta \cdot \mathrm{vol}(\lat)^{1/n}$.

$\vec{B}$ is {\em $\delta$-DHSVP-reduced} \cite{GNFindingShort08, ALNSSlideReduction19} (where D
stands for dual) if the reversed dual basis $\vec{B}^{-s}$ is  $\delta$-HSVP-reduced and it implies that 
\[ \mathrm{vol}(\lat)^{1/n}  \leq \delta \cdot \|\vec{b}_n^{\ast}\| \; . \]

Given a $\delta$-(H)SVP oracle on lattices with rank at most $n$, we can efficiently compute a $\delta$-(H)SVP-reduced
basis or a $\delta$-D(H)SVP-reduced basis for any rank $n$ lattice $\lat \subseteq \Z^m$. 
Furthermore, this also applies for a projected block of basis. More specifically, with access to a $\delta$-(H)SVP oracle for lattices with rank at most $k$, given any basis $\basis = (\vec{b}_1,\ldots, \vec{b}_n) \in \Z^{m\times n}$ of $\lat$ and an index $i \in [1,n-k+1]$, we can efficiently compute a size-reduced basis $\mathbf{C} = (\vec{b}_1,\ldots, \vec{b}_{i-1}, \vec{c}_i,\ldots, \vec{c}_{i+k-1}, \vec{b}_{i+k},\ldots, \vec{b}_n)$ such that $\mathbf{C}$ is a basis for $\lat$ and the projected block $\mathbf{C}_{[i,i+k-1]}$ is $\delta$-(H)SVP-reduced or $\delta$-D(H)SVP reduced. Moreover, we note the following:
\begin{itemize}
	\item If $\mathbf{C}_{[i,i+k-1]}$ is $\delta$-(H)SVP-reduced, %
	the procedures in \cite{GNFindingShort08,MWPracticalPredictable16} equipped with $\delta$-(H)SVP-oracle ensure that $\|\mathbf{C}^{\ast}\|\le \|\basis^{\ast}\|$;
	\item If $\mathbf{C}_{[i,i+k-1]}$ is $\delta$-D(H)SVP-reduced, the inherent LLL reduction implies $\|\mathbf{C}^{\ast}\|\le 2^{k}\|\basis^{\ast}\|$. Indeed, the GSO %
	of $\mathbf{C}_{[i,i+k-1]}$ satisfies \[\|(\mathbf{C}_{[i,i+k-1]})^{\ast}\|\le 2^{k/2}\lambda_{k}(\lat(\mathbf{C}_{[i,i+k-1]}))\] (by \cite[p.\,518, Line 27]{LLLFactoringPolynomials82}) and $\lambda_{k}(\lat(\mathbf{C}_{[i,i+k-1]}))\le \sqrt{k}\|\basis^{\ast}\|$. Here, $\lambda_k(\cdot)$ denotes the $k$-th minimum. 
\end{itemize}

Therefore, with size reduction, performing $\poly(n, \log \|\basis\|)$ many such operations will increase $\|\basis^{\ast}\|$ and hence $\|\basis\|$ by at most a factor of $2^{\poly(n,\log\|B\|)}$.
If the number of operations is bounded by $\poly(n, \log \|\basis\|)$, all intermediate steps and the total running time (excluding oracle queries) will be polynomial in the initial input size; Details can be found in e.g., \cite{GNFindingShort08,LNApproximatingDensest14}.
Hence, we will focus on bounding the number of calls to such block reduction subprocedures when we analyze the running time of basis reduction algorithms.

\paragraph{Twin reduction}
The following notion of twin reduction and the subsequent fact comes from \cite{GNFindingShort08, ALNSSlideReduction19}.

A basis $\vec{B} = (\vec{b}_1,\ldots, \vec{b}_{d+1})$  is $\delta$-{\em twin-reduced} if $\basis_{[1,d]}$ is $\delta$-{\em HSVP-reduced} and $\basis_{[2,d+1]}$ is $\delta$-{\em DHSVP-reduced}.

\begin{fact}
	\label{fact:twinsies}
	If $\basis := (\vec{b}_1,\ldots, \vec{b}_{d+1}) \in \R^{m \times (d+1)}$ is $\delta$-twin-reduced, then 
	\begin{equation}
	\label{eq:twin_decay}
	\|\vec{b}_1\| \le \delta^{2d/(d-1)} \|\vec{b}^*_{d+1}\|
	\; .
	\end{equation}
\end{fact}

\subsection{The DBKZ algorithm}

We augment Micciancio and Walter's elegant DBKZ algorithm \cite{MWPracticalPredictable16} with a $\delta_H$-HSVP-oracle instead of an SVP-oracle since the SVP-oracle is used as a $\sqrt{\gamma_{k}}$-HSVP oracle everywhere in their algorithm. See \cite{ALNSSlideReduction19} for a high-level sketch of the proof. %

\begin{algorithm}
	\small
	\caption{%
		The Micciancio-Walter  DBKZ %
		algorithm \cite[Algorithm 1]{MWPracticalPredictable16}\label{alg:SDBKZ}}
	\begin{algorithmic}[1]
		\REQUIRE A block size $k \ge 2$, number of tours $N$,  
		a basis $\basis = (\mathbf{b}_{1}, \cdots,
		\mathbf{b}_{n}) \in \Z^{m \times n}$, and access to a $\delta_H$-HSVP oracle for lattices with rank $k$.

		\ENSURE  A new basis of $\lat(\basis)$.

		\FOR{$\ell = 1$ \TO $N$}  %
		\label{step of SDBKZ alg:repeat}

		\FOR{$i = 1$ \TO $n-k$}\label{step of SDBKZ alg:for}
		
		\STATE %
		$\delta_H$-HSVP-reduce $\basis_{[i,i+k-1]}$.
		\label{step of LSDBKZ alg:size-reduction}
		\ENDFOR %
		\label{step of SDBKZ alg:firt end for}
		\FOR{$j = n-k+1$ \TO $1$} \label{step of SDBKZ alg:for2}

		\STATE %
		$\delta_H$-DHSVP-reduce $\basis_{[j,j+k-1]}$%

		\ENDFOR %
		\label{step of SDBKZ alg:end for}

		\ENDFOR
		\label{step of SDBKZ alg:until}
		\STATE  $\delta_H$-HSVP-reduce $\basis_{[1,k]}$.
		\RETURN $\basis$. \label{step of SDBKZ alg:return}
	\end{algorithmic}
\end{algorithm}

\begin{theorem}\label{th:SDBKZ}
	For integers $n > k \geq 2$, an approximation factor $1 \leq \delta_H \leq 2^k$,
	an input basis $\basis_{0} \in \Z^{m \times n}$ for a lattice $\lat  \subseteq \Z^m$, and 
	$
	N := \lceil (2n^2/(k-1)^2) \cdot \log (n\log(5\|\basis_{0}\|)/\eps) \rceil$
	for some $\eps \in [2^{-\poly(n)},1]$, %
	Algorithm \ref{alg:SDBKZ} outputs a
	basis $\basis$ of $\lat$ in polynomial time (excluding oracle queries) such that
	\[
	\|\vec{b}_1\| \leq (1+\eps)\cdot (\delta_H)^{\frac{n-1}{(k-1)}}\vol(\lat)^{1/n}
	\; ,
	\]
	by making $N
	\cdot (2n-2k+1)+1$ calls to the $\delta_H$-HSVP oracle for lattices with rank $k$.
\end{theorem}

\section{Smooth sublattices and \texorpdfstring{$\bareta_\eps(\lat)$}{eta bar}}

 The analysis of our algorithm relies on the existence of a \emph{smooth sublattice} $\lat' \subseteq \lat$ of our input lattice $\lat \subset \R^n$, i.e., a sublattice $\lat'$ such that $\eta_\eps(\lat')$ is small (relative to, say, $\lambda_1(\lat)$ or $\det(\lat)^{1/n}$). To that end, for $\eps > 0$ and a lattice $\lat \subset \R^n$, we define
\[
    \bareta_\eps(\lat) := \min_{\lat' \subseteq \lat} \eta_\eps(\lat')
    \; ,
\]
where the minimum is taken over all sublattices $\lat' \subseteq \lat$. (It is not hard to see that the minimum is in fact achieved. Notice that any minimizer $\lat'$ must be a primitive sublattice, i.e., $\lat' = \lat \cap \spn(\lat')$.)

We will now prove that $\bareta_\eps(\lat)$ is bounded both in terms of $\lambda_1(\lat)$ and $\det(\lat)$.

\begin{lemma}
\label{lem:bar_eta}
    For any lattice $\lat \subset \R^n$ and any $\eps \in (0,1/2)$, 
    \[
        \lambda_1(\lat)/\sqrt{n} \leq \bareta_\eps(\lat) \leq \sqrt{\log(1/\eps)} \cdot \min\{ \lambda_1(\lat), 10(\log n + 2)\det(\lat)^{1/n} \} 
        \; . 
    \]
\end{lemma}

The bounds in terms of $\lambda_1(\lat)$ are more-or-less trivial.
The bound $\bareta_\eps(\lat) \lesssim \sqrt{\log(1/\eps) \log n} \det(\lat)^{1/n}$ follows from the main result in \cite{Regev:2017:RMT:3055399.3055434} (originally conjectured by Dadush~\cite{DRStrongReverse16}), which is called a ``reverse Minkowski theorem'' and which we present below. (In fact, Lemma~\ref{lem:bar_eta} is essentially equivalent to the main result in \cite{Regev:2017:RMT:3055399.3055434}.)

\begin{definition}
	A lattice $\lat \subset \R^n$ is a \textbf{stable} lattice if $\det(\lat) = 1$ and $\det(\lat') \geq 1$ for all lattices $\lat' \subseteq \lat$.
\end{definition}
\begin{theorem}[\cite{Regev:2017:RMT:3055399.3055434}]\label{thm:stable_lattice}
	For any stable lattice $\lat \subset \R^n$, $\eta_{1/2}(\lat) \leq 10(\log n + 2)$.
\end{theorem}

\begin{proof}[Proof of Lemma~\ref{lem:bar_eta}]
    The lower bound on $\bareta_\eps(\lat)$ follows immediately from Theorem~\ref{thm:eta_lambda_n} together with the fact that $\lambda_1(\lat) \leq \lambda_1(\lat') \leq \lambda_n(\lat')$ for any sublattice $\lat' \subseteq \lat$. The bound $\bareta_\eps(\lat) \leq \sqrt{\log(1/\eps)} \cdot \lambda_1(\lat)$ is immediate from Claim~\ref{clm:Z_smooth} applied to the one-dimensional lattice $\Z \vec{v}$ generated by $\vec{v} \in \lat$ with $\|\vec{v}\| = \lambda_1(\lat)$.
    
    So, we only need to prove that $\bareta_{1/2}(\lat) \leq 10(\log n + 2)\det(\lat)^{1/n}$. The result for all $\eps \in (0,1/2)$ then follows from Corollary~\ref{cor:smoothing_parameter_scaled}.
    
    We prove this by induction on $n$. The result is trivial for $n = 1$. (Indeed, for $n = 1$ we have $\det(\lat)^{1/n} = \lambda_1(\lat)$.) 
    For $n > 1$, we first assume without loss of generality that $\det(\lat) = 1$. If $\lat \subset \R^n$ is stable, then the result follows immediately from Theorem~\ref{thm:stable_lattice}. Otherwise, there exists a sublattice $\lat' \subset \lat$ such that $\det(\lat') < 1$. Notice that $k := \rank(\lat') < n$. Therefore, by the induction hypothesis, $\bareta_{1/2}(\lat') \leq 10(\log k + 2) \det(\lat')^{1/k} < 10 (\log n + 2)$. The result then follows from the fact that $\bareta_\eps(\lat) \leq \bareta_\eps(\lat')$ for any sublattice $\lat' \subseteq \lat$.
\end{proof}

\subsection{Sampling with parameter \texorpdfstring{$\poly(n) \cdot \bareta_\eps(\lat)$}{poly(n)*(bar eta)}}

\begin{lemma}
\label{lem:HKZ_sublattice_still_smooth}
    For any lattice $\lat \subset \R^n$, $\gamma \geq 1$, $\eps \in (0,1/2)$, $\gamma$-HKZ-reduced basis $\basis = (\vec{b}_1,\ldots, \vec{b}_n)$ of $\lat$, $\eps \in (0,1/2)$, and index $i \in \{2,\ldots, n\}$ such that 
    \[
    \|\vec{b}_i^*\| > \gamma \sqrt{n} \cdot \bareta_\eps(\lat)
    \; ,
    \]
    we have
    $
          \bareta_\eps(\lat(\vec{b}_1,\ldots, \vec{b}_{i-1})) = \bareta_\eps(\lat)\; .
     $
\end{lemma}
\begin{proof}
Suppose that $\lat' \subseteq \lat$ satisfies $\eta_\eps(\lat') = \bareta_\eps(\lat) < \|\vec{b}_i^*\|/(\gamma \sqrt{n})$ with $k := \rank(\lat')$. We wish to show that $\lat' \subseteq \lat(\vec{b}_1,\ldots, \vec{b}_{i-1})$, or equivalently, that $\pi_i(\lat') = \{\vec0\}$. 
    Indeed, by Theorem~\ref{thm:eta_lambda_n}, $\lambda_k(\lat') \leq \sqrt{k}\cdot \eta_\eps(\lat') \leq \sqrt{n} \cdot \bareta_\eps(\lat)$. In particular, there exist $\vec{v}_1,\ldots, \vec{v}_k \in \lat'$ with $\spn(\vec{v}_1,\ldots, \vec{v}_k) = \spn(\lat')$ and 
    \[
        \|\pi_i(\vec{v}_j)\| \leq \|\vec{v}_j\| \leq \lambda_k(\lat') \leq \sqrt{n} \cdot \bareta_\eps(\lat) <  \|\vec{b}_i^*\|/\gamma
        \; 
    \]
    for all $j \in \{1,\ldots, k\}$.
    Therefore, if $\pi_i(\vec{v}_j) \neq \vec0$. Then, $\pi_i(\vec{v}_j) \in \pi_i(\lat)$ is a non-zero vector with norm strictly less than $\|\vec{b}_i^*\|/\gamma$, which implies that $\lambda_1(\pi_i(\lat)) < \|\vec{b}_i^*\|/\gamma$, contradicting the assumption that $\basis$ is a $\gamma$-HKZ basis. Therefore, $\pi_i(\vec{v}_j) = \vec0$ for all $j$, which implies that $\pi_i(\lat') = \{\vec0\}$, i.e., $\lat' \subseteq \lat(\vec{b}_1,\ldots, \vec{b}_{i-1})$, as needed.
\end{proof}

\begin{prop}
\label{prop:start}
    There is a $(2^{r + o(r)} + M)\cdot \poly(n, \log M)$-time algorithm that takes as input a (basis for a) lattice $\lat \subset \R^n$, $2 \leq r \leq n$, an integer $M \geq 1$, and a parameter 
    \[
        s \geq r^{n/r} \sqrt{n \log n} \cdot \bareta_{\eps}(\lat)
        \]
        for some $\eps \in (0,1/2)$
    and outputs a (basis for a) sublattice $\widehat{\lat} \subseteq \lat$ with $\bareta_{\eps}(\widehat{\lat}) = \bareta_\eps(\lat)$
    and $\vec{X}_1,\ldots, \vec{X}_M \in \widehat{\lat}$ that are sampled independently from $D_{\widehat{\lat}, s}$.
\end{prop}
\begin{proof}
    The algorithm takes as input a (basis for a) lattice $\lat \subset \R^n$, $2 \leq r \leq n$, $M \geq 1$, and a parameter $s > 0$ and behaves as follows. It first uses the procedure from Theorem~\ref{thm:BKZ} to compute a $\gamma$-HKZ reduced basis $\vec{b}_1,\ldots, \vec{b}_n$, where $\gamma := r^{n/r}$. Let $i \in \{1,\ldots, n\}$ be maximal such that $\|\vec{b}_j^*\| \leq s/\sqrt{\log n} $ for all $j \leq i$, and let $\widehat{\lat} := \lat(\vec{b}_1,\ldots, \vec{b}_i)$. (If no such $i$ exists, the algorithm simply fails.) The algorithm then runs the procedure from Theorem~\ref{thm:DG_initialisation} repeatedly  to sample $\vec{X}_1,\ldots, \vec{X}_M \sim D_{\widehat{\lat},s}$ and outputs $\widehat{\lat}$ and $\vec{X}_1,\ldots, \vec{X}_M$.
    
    The running time of the algorithm is clearly $(2^{10r}+M) \cdot \poly(n, \log M)$. By Theorem~\ref{thm:DG_initialisation}, the $\vec{X}_i$ have the correct distribution. Notice that, if the algorithm fails, then 
    \[
        \|\vec{b}_1\| > s/\sqrt{\log n} \geq \gamma \sqrt{n} \cdot \bareta_{\eps}(\lat)
        \; .
    \]
    Recalling that $\|\vec{b}_1\| \leq \gamma \lambda_1(\lat)$, it follows that $\sqrt{n} \bareta_\eps(\lat) < \lambda_1(\lat)$, which contradicts Lemma~\ref{lem:bar_eta}. So, the algorithm never fails (provided that the promise on $s$ holds).
    
    It remains to show that $\bareta_\eps(\lat) = \bareta_\eps(\lat(\vec{b}_1,\ldots, \vec{b}_i))$. If $i = n$, then this is trivial. Otherwise, $i \in \{1,\ldots, n-1\}$, and we have 
    \[
        \|\vec{b}_{i+1}^*\| > s/\sqrt{\log n} \geq \gamma \sqrt{n} \cdot \bareta_{\eps}(\lat)
        \; .
    \]
    The result follows immediately from
    Lemma~\ref{lem:HKZ_sublattice_still_smooth}.
\end{proof}\section{An approximation algorithm for HSVP and SVP}
In this section, we present our algorithm that solves $\widetilde{O}(\sqrt{n})$-HSVP and $\widetilde{O}(\sqrt{n})$-SVP in $2^{n/2+o(n)}$ time. More precisely, we provide a detailed analysis of a simple ``pair-and-sum'' algorithm, which will solve $O(\sqrt{n}) \cdot\bareta_\eps(\lat)$-GSVP for $\eps = 1/\poly(n)$. This in particular yields an algorithm that simultaneously solves $\widetilde{O}(\sqrt{n})$-SVP and $\widetilde{O}(\sqrt{n})$-HSVP.

\subsection{Mixtures of Gaussians}

We will be working with random variables $\vec{X}$ that are ``mixtures'' of discrete Gaussians, i.e., random variables that can be written as $D_{\lat + \vec{C}, s}$ for some lattice $\lat \subset \R^n$, parameter $s > 0$, and random variable $\vec{C} \in \R^n$. In other words, $\vec{X}$ can be sampled by first sampling $\vec{C} \in \R^n$ from some arbitrary distribution and then sampling $\vec{X}$ from $D_{\lat + \vec{C}, s}$. E.g., the discrete Gaussian $D_{\lat, s}$ itself is such a distribution, as is the discrete Gaussian $D_{\widehat{\lat}, s}$ for any superlattice $\widehat{\lat} \supseteq \lat$.
Indeed, in our applications, we will always have $\vec{C} \in \widehat{\lat}$ for some superlattice $\widehat{\lat} \supseteq \lat$, and we will initialize our algorithm with samples from $D_{\widehat{\lat},s}$.

Our formal definition below is a bit technical, since we must consider the joint distribution of many such random variables that are only $\delta$-similar to these distributions and satisfy a certain independence property. %
In particular, we will work with $\vec{X}_1,\ldots, \vec{X}_M$ such that each $\vec{X}_i$ is $\delta$-similar to $\vec{Y}_i \sim D_{\lat + \vec{C}_i, s}$, where $\vec{C}_i$ is an arbitrary random variable (that might depend on the $\vec{X}_j$) but once $\mathbf{C}_i$ is fixed, $\vec{Y}_i$ is sampled from $D_{\lat + \vec{C}_i, s}$ independently of everything else. Here and below, we adopt the convention that $\Pr[A\  |\ B] = 0$ whenever $\Pr[B] = 0$, i.e., all probabilities are zero when conditioned on events with probability zero.

\begin{definition}
For (discrete) random variables $\vec{X}_1,\ldots, \vec{X}_m \in \R^n$ and $i \in \{1,\ldots, m\}$, let $\vec{X}_{-i} := (\vec{X}_1, \ldots, \vec{X}_{i-1},\vec{X}_{i+1},\ldots, \vec{X}_m) \in \R^{(m-1)n}$. We say that $\vec{X}_1, \ldots, \vec{X}_m$ are $\delta$-similar to a mixture of independent Gaussians over $\lat$ with parameter $s > 0$ if for any $i \in \{1,\ldots, m\}$, $\vec{y} \in \R^n$, and $\vec{w} \in \R^{(m-1)n}$,
\[
    \Pr[\vec{X}_i = \vec{y} \ | \ \vec{X}_{-i} = \vec{w}] = e^{\pm \delta} \cdot \frac{\rho_s(\vec{y})}{\rho_{s}(\lat + \vec{y})} \cdot \Pr[\vec{X}_i \in \lat + \vec{y} \ | \ \vec{X}_{-i} = \vec{w}]
    \; .
\]
\end{definition}

Additionally we will need the distribution we obtain at every step to be symmetric about the origin as defined below.
\begin{definition}
	 We say that a list of (discrete) random variables $\vec{X}_1,\ldots, \vec{X}_m \in \R^n$ is symmetric if for any $i \in \{1,\ldots, m\}$, any $\vec{y} \in \R^n$, and any $\vec{w} \in \R^{(m-1)n}$, 
	\[\Pr[\vec{X}_i = \vec{y} \ | \ \vec{X}_{-i} = \vec{w}] = \Pr[\vec{X}_i = -\vec{y} \ | \ \vec{X}_{-i} = \vec{w}]\; .\]
\end{definition}

 We need the following simple lemma that bounds the probability of $\vec{X}$ being $\vec{0}$, where $\vec{X}$ is distributed as a mixture of discrete Gaussians over $\cL$. 

\begin{lemma}\label{lem:zero_vector_weight}
For any lattice $ \lat \subset \R^n$, let $\vec{X}_1,\ldots, \vec{X}_m \in \cL$ be $\delta$-similar to a mixture of independent Gaussians over $\lat$ with parameter $s \geq \beta\eta_{1/2}(\lat)$ for some $\beta > 1$. Then, for any $i$, and any $\vec{w} \in \R^{(m-1) n}$
	\[\Pr[\vec{X}_i = \vec{0} \ | \ \vec{X}_{-i} = \vec{w}] \leq \frac{3e^\delta}{2\beta}\; .\]
\end{lemma}

\begin{proof}
Let $s' := \eta_{1/2}(\lat)$. We have that
	\[\Pr[\vec{X}_i = \vec{0} \ | \ \vec{X}_{-i} = \vec{w}] \leq \Pr[\vec{X}_i = \vec{0} \ | \ \vec{X}_i \in \lat,\ \vec{X}_{-i} = \vec{w}]\leq \frac{e^\delta}{\rho_s(\lat)} \leq e^{\delta} \cdot \frac{\rho_{s'}(\lat)}{\rho_s(\lat)} \; .\]
	The result then follows from Claim~\ref{clm:smooth_mass_growth}.
\end{proof}

The following corollary shows that a mixture of discrete Gaussians must contain a short non-zero vector in certain cases.

\begin{cor}\label{co:short_vector}
    For any lattices $\lat' \subseteq \lat \subset \R^n$, parameter $s \geq 10e^\delta \eta_{1/2}(\lat')$, $m \geq 100$, and random variables $\vec{X}_1,\ldots, \vec{X}_m$ that are $\delta$-similar to mixtures of independent Gaussians over $\lat'$ with parameter $s$,
    \[
        \Pr[\exists i \in [1,m] \text{ such that } 0 < \|\vec{X}_i\|^2 < 4T] \geq 1/10
        \; ,
    \]
    where $T := \frac{1}{m} \sum_{i=1}^m \expect[\|\vec{X}_i\|^2]$.
\end{cor}
\begin{proof}
	By Markov's inequality, we have
	\[\Pr\Big[\sum_{i = 1}^m \|\vec{X}_i \|^2 \geq 2 mT\Big] \leq \frac{ 1}{2} \; .\]
	Hence, with probability at least $\frac{1}{2}$, we have $\sum_{i = 1}^m \|\vec{X}_i \|^2 < 2mT$. 
	
	We next note that many of the $\vec{X}_i$ must be non-zero with high probability. Let $Y_1,\ldots, Y_m \in \{0,1\}$ such that $Y_i = 0$ if and only if $\vec{X}_i = \vec0$. By Lemma~\ref{lem:zero_vector_weight}, 
	\[
	    \expect[Y_i \ | \ Y_1 = y_1,\ldots, Y_{i-1} = y_{i-1}] \geq 4/5
	\]
	for any $y_1,\ldots, y_{i-1} \in \{0,1\}$. By Corollary~\ref{cor:Azuma}, we have that \[\Pr[Y_1 + \cdots + Y_m \leq 3m/5] \leq e^{-m/100} \leq 1/e \;. \]
	
	Finally, by union bound, we see that with probability at least $1- 1/e - 1/2 > 1/10$ the average squared norm will be at most $2T$ and more than half of the $\vec{X}_i$ will be non-zero. It follows from another application of Markov's inequality that at least one of the non-zero $\vec{X}_i$ must have squared norm less than $4T$.
\end{proof}

\subsection{Summing vectors}

Our algorithm will start with vectors $\vec{X}_1,\ldots, \vec{X}_m \in \lat_0$, where $\lat_0 \subset \lat$ is some very dense sublattice of the input lattice $\lat$. It then takes sums $\vec{Y}_k = \vec{X}_i + \vec{X}_j$ of pairs of these in such a way that the resulting $\vec{Y}_k$ lie in some appropriate sublattice $\lat_1 \subset \lat_0$, i.e., $\vec{Y}_k \in \lat_1$. It does this repeatedly, finding vectors in $\lat_2, \lat_3,\ldots, \lat_\ell$ until finally it obtains vectors in $\lat_\ell := \lat$. 

Here, we study a single step of this algorithm, as shown below. %

\begin{algorithm}
\small
\caption{One step of the algorithm. \label{alg:one_step}}
\begin{algorithmic}[1]
\REQUIRE Lattices $\lat_0, \lat_1 \subset \R^n$ with $2\lat_0 \subseteq \lat_1 \subseteq \lat_0$, and lattice vectors $\vec{X}_1,\ldots, \vec{X}_m \in \lat_0$ with $m \geq 2|\lat_0/\lat_1|$.

\ENSURE  Lattice vectors $\vec{Y}_1,\ldots, \vec{Y}_M \in \lat_1$, with $M := \ceil{(m - |\lat_0/\lat_1|)/2}$.

\STATE Set $\mathsf{USED}_i := \FALSE$ for $i = 1,\ldots, m$, $k = 1$, and $i = 1$.

\WHILE{$k \leq M$}

\IF{\NOT $\mathsf{USED}_i$ \AND ($\exists j \in \{1,\ldots, m\} \setminus \{i\}$ such that $\vec{X}_j \equiv \vec{X}_i \bmod \lat_1$ \AND $\mathsf{USED}_j = \FALSE$)}

\STATE Let $j \neq i$ be minimal such that $\vec{X}_j \equiv \vec{X}_i \bmod \lat_1$ \AND $\mathsf{USED}_j = \FALSE$.

\STATE Set $\vec{Y}_k = \vec{X}_i + \vec{X}_j$.

\STATE Set $\mathsf{USED}_i = \mathsf{USED}_j = \TRUE$ and increment $k$.

\ENDIF

\STATE Increment $i$.

\ENDWHILE

\RETURN{$\vec{Y}_1,\ldots, \vec{Y}_M$}

\end{algorithmic}
\end{algorithm}

Notice that Algorithm~\ref{alg:one_step} can be implemented in time $m \cdot \poly(n, \log m)$. This can be done, e.g., by creating a table of the $\vec{X}_i$ sorted according to $\vec{X}_i \bmod \lat_1$. Then, for each $i$, such a $j$ can be found (if it exists) by performing binary search on the table. Furthermore, the algorithm is guaranteed to find $M = \ceil{(m - |\lat_0/\lat_1|)/2}$ output vectors because at most $|\lat_0/\lat_1|$ of the input vectors can be unpaired.

The key property that we will need from Algorithm~\ref{alg:one_step} is that for any (possibly unknown) sublattice $\lat' \subseteq \lat_1 \subseteq \lat_0$, the algorithm maps mixtures of Gaussians over $\lat'$ to mixtures of Gaussians over $\lat'$, provided that the parameter $s$ is significantly above $\eta_\eps(\lat')$. In other words, as long as there exists some sublattice $\lat' \subseteq \lat_1$ such that $\eta_\eps(\lat') \lesssim s$, then the output of the algorithm will be a mixture of Gaussians. Indeed, this is more-or-less immediate from Lemma~\ref{lem:convolution}.

\begin{lemma}
\label{lem:mixture_to_mixture}
    For any lattices $\lat_0, \lat_1, \lat' \subset \R^n$ with $2\lat_0 \subseteq \lat_1 \subseteq \lat_0$ and $\lat' \subseteq \lat_1$, $\eps \in (0,1/3)$, $\delta > 0$, and parameter $s \geq \sqrt{2} \eta_\eps(\lat')$,
    if the input vectors $\vec{X}_1,\ldots, \vec{X}_m \in \lat_0$ are sampled from the distribution that is $\delta$-similar to a mixture of independent Gaussians over $\lat'$ with parameter $s$, then the output vectors $\vec{Y}_1,\ldots, \vec{Y}_M \in \lat_1$ are $(2\delta + 3\eps)$-similar to a mixture of independent Gaussians over $\lat'$ with parameter $\sqrt{2} s$.
\end{lemma}
\begin{proof}
For a list of cosets $\vec{d} := (\vec{c}_1, \ldots, \vec{c}_m) \in (\lat_0/\lat')^m$ such that $\Pr[\vec{X}_1 = \vec{c}_1 \bmod \lat', \ldots, \vec{X}_m = \vec{c}_m \bmod \lat']$ is non-zero, let $\vec{Y}_{\vec{d}, 1},\ldots, \vec{Y}_{\vec{d}, M}$ be the random variables obtained by taking $\vec{Y}_1,\ldots, \vec{Y}_M$ conditioned on $\vec{X}_i \equiv \vec{c}_i \bmod \lat'$ for all $i$. We similarly define $\vec{X}_{\vec{d}, i}$. Notice that $\vec{Y}_1,\ldots, \vec{Y}_M$ is a convex combination of random variables of the form $\vec{Y}_{\vec{d}, 1},\ldots, \vec{Y}_{\vec{d}, M}$, and that the property of being close to a mixture of independent Gaussians is preserved by taking convex combinations. Therefore, it suffices to prove the statement for $\vec{Y}_{\vec{d}, 1},\ldots, \vec{Y}_{\vec{d}, M}$ for all fixed $\vec{d}$.

To that end, fix $k \in \{1,\ldots, M\}$ and such a $\vec{d} \in (\lat_0/\lat')^m$. 
Notice that $\vec{X}_{\vec{d}, i} \in \lat' + \vec{c}_i \subseteq \lat_1 + \vec{c}_i$. Therefore, there exist fixed $i,j$ such that $\vec{Y}_{\vec{d},k} = \vec{X}_{\vec{d}, i} + \vec{X}_{\vec{d}, j}$. Furthermore, by assumption, for any $\vec{w} \in \lat_0^{m-1}$ and $\vec{x} \in \lat_0$,
\[
    \Pr[\vec{X}_{\vec{d},i} = \vec{x} \ | \ \vec{X}_{\vec{d}, -i} = \vec{w}] = e^{\pm \delta} \frac{\rho_{s}(\vec{x})}{\rho_s(\lat' + \vec{c}_i)}
    \; ,
\]
and likewise for $j$.
It follows from Lemma~\ref{lem:convolution} that for any $\vec{y} \in \lat_1$ and $\vec{z} \in \lat_1^{M-1}$,
\[
    \Pr[\vec{X}_{\vec{d},i} + \vec{X}_{\vec{d}_j} = \vec{y} \ | \ \vec{Y}_{\vec{d}, -k}  = \vec{z}] = e^{\pm (2\delta + 3\eps)} \frac{\rho_{\sqrt{2} s}(\vec{y})}{\rho_{\sqrt{2} s}(\lat' + \vec{c}_i + \vec{c}_j)}
    \; ,
\]
as needed.
\end{proof}

\begin{lemma}
\label{lem:symmetric}
    For any lattices $\lat_0, \lat_1 \subset \R^n$ with $2\lat_0 \subseteq \lat_1 \subseteq \lat_0$, if the input vectors $\vec{X}_1,\ldots, \vec{X}_m \in \lat_0$ are sampled from a symmetric distribution, then the distribution of the output vectors $\vec{Y}_1,\ldots, \vec{Y}_M$ will also be symmetric. Furthermore, 
    \[
        \sum \expect[\|\vec{Y}_k\|^2] \leq \sum \expect[\|\vec{X}_i\|^2] \; .
    \]
\end{lemma}
\begin{proof}
    Let $\vec{d} = (\vec{c}_1,\ldots, \vec{c}_m) \in (\lat_0/\lat_1)^m$ be a list of cosets such that with non-zero probability we have $\vec{X}_1 \in \lat_1 + \vec{c}_1,\ldots, \vec{X}_m \in \lat_1 + \vec{c}_m$. Let $\vec{X}_{\vec{d},1},\ldots, \vec{X}_{\vec{d}, m}$ be the distribution obtained by sampling the $\vec{X}_i$ conditioned on this event, and let $\vec{Y}_{\vec{d},1},\ldots, \vec{Y}_{\vec{d}, M}$ be the corresponding output.
    
    Notice that the distribution of $\vec{X}_{\vec{d},1},\ldots, \vec{X}_{\vec{d},m}$ is also symmetric, since $\lat_1 + \vec{c} = -(\lat_1 + \vec{c})$ for any $\vec{c} \in \lat_0/\lat_1$. (Here, we have used the fact that $2\lat_0 \subseteq \lat_1 \subseteq \lat_0$.) 
    
    And, for fixed $\vec{d}$ and $k \in \{1,\ldots, M\}$ there exist fixed (distinct) $i,j \in \{1,\ldots, m\}$ such that
    \[
        \vec{Y}_{\vec{d},k} = \vec{X}_{\vec{d}, i} + \vec{X}_{\vec{d},j}
        \; .
    \]
    But, since the $\vec{X}_{\vec{d}, 1},\ldots, \vec{X}_{\vec{d},m}$ are distributed symmetrically, we see immediately that for any $\vec{y} \in \lat_1$ and $\vec{w} \in \lat_1^{M-1}$,
    \[
        \Pr[\vec{Y}_{\vec{d}, k} = \vec{y} \ | \ \vec{Y}_{\vec{d}, -k} = \vec{w}] = \Pr[\vec{Y}_{\vec{d}, k} = -\vec{y} \ | \ \vec{Y}_{\vec{d}, -k} = \vec{w}] \; .
    \]
    In other words, the distribution of $\vec{Y}_{\vec{d},1},\ldots, \vec{Y}_{\vec{d}, M}$ is symmetric.
    
    Furthermore,
    \[
        \expect[\|\vec{X}_{\vec{d}, i} + \vec{X}_{\vec{d}, j}\|^2] = \expect[\|\vec{X}_{\vec{d}, i}\|^2] + \expect[\|\vec{X}_{\vec{d}, j}\|^2] + 2 \expect[\langle \vec{X}_i, \vec{X}_j\rangle] = \expect[\|\vec{X}_{\vec{d}, i}\|^2] + \expect[\|\vec{X}_{\vec{d}, j}\|^2]
        \; ,
    \]
    where in the last step we have used the symmetry of $\vec{X}_{\vec{d}, 1},\ldots, \vec{X}_{\vec{d}, m}$. Since the $\vec{Y}_{\vec{d},k}$ are sums of disjoint pairs of the $\vec{X}_{\vec{d}, i}$, it follows immediately that
    \[
        \sum_{k=1}^M \expect[\|\vec{Y}_{\vec{d}, k}\|^2] \leq \sum_{i=1}^m \expect[\|\vec{X}_{\vec{d}, i}\|^2]
        \; .
    \]
    
    The results for $\vec{X}_{1},\ldots, \vec{X}_m, \vec{Y}_1,\ldots, \vec{Y}_M$ then follow immediately from the fact that this distribution can be written as a convex combination of $\vec{X}_{\vec{d},1},\ldots, \vec{X}_{\vec{d},m},\vec{Y}_{\vec{d},1},\ldots, \vec{Y}_{\vec{d},M}$ for different coset lists $\vec{d} \in (\lat_0/\lat_1)^m$, since both symmetry and the inequality on expectations are preserved by convex combinations.
\end{proof}

\subsection{A tower of lattices}\label{sec:tower}
We will repeatedly apply Algorithm \ref{alg:one_step} on a ``tower" of lattices similar to \cite{ADRSSolvingShortest15}. We use (a slight modification of) the definition and construction of the tower of lattices from \cite{ADRSSolvingShortest15}.
\begin{definition}[\cite{ADRSSolvingShortest15}] %
	For an integer $\alpha$ satisfying $n/2 \leq \alpha \leq n$, we say that $(\lat_0, \ldots, \lat_\ell)$ is a tower of lattices in $\R^n$ of index $2^\alpha$ if for all $i$ we have $2\lat_{i - 1} \subseteq \lat_{i} \subset \lat_{i-1}, \lat_{i} / 2 \subseteq \lat_{i - 2}$, $|\lat_{i-1}/\lat_i| = 2^\alpha$, and $2^{\ceil{i\alpha/n}} \lat_0 \subseteq \lat_i \subseteq 2^{\floor{i\alpha/n}}\lat_0$ for all $i$.
\end{definition}
\begin{theorem}[\cite{ADRSSolvingShortest15}]\label{thm:tower_of_lattices} %
	There is a polynomial-time algorithm that takes as input integers $\ell\geq 1$ and $n/2 \leq \alpha\leq n$ as well as a lattice ${\lat} \subseteq \R^n$ and outputs a tower of lattice $(\lat_0, \ldots, \lat_\ell)$ with $\lat_\ell = {\lat}$. 
\end{theorem}
\begin{proof}
We give the construction below. The desired properties are immediate from the construction. 
	Let $\vec{b}_1, \ldots, \vec{b}_n$ be a basis of ${\lat}$. The tower is then defined by ``cyclically halving $\alpha$ coordinates'', namely,
	\begin{align*}
	\lat_\ell &= {\lat}(\vec{b}_1, \ldots, \vec{b}_n), \\
	\lat_{\ell-1} &= {\lat}(\vec{b}_1/2, \ldots, \vec{b}_{\alpha}/2, \vec{b}_{\alpha + 1}, \ldots \vec{b}_n), \\
	\lat_{\ell-2} &= {\lat}(\vec{b}_1/4, \ldots, \vec{b}_{2\alpha-n}/4, \vec{b}_{2\alpha - n+ 1}/2, \ldots \vec{b}_n/2), 
	\end{align*}
	etc. The required properties can be easily verified.
\end{proof}

The following proposition shows that starting with discrete Gaussian samples from $\cL_0$ and then repeatedly applying Algorithm~\ref{alg:one_step} gives us a list of vectors in $\cL_\ell$ that is close to a mixture of Gaussians, provided that there exists an appropriate ``smooth sublattice'' $\lat' \subseteq \lat_\ell$.
\begin{prop}\label{prop:multi_steps}
    There is an algorithm that runs in $m \cdot \poly(n,\ell, \log m)$ time; takes as input a tower of lattices $(\lat_0, \ldots, \lat_\ell)$ in $\R^n$ of index $2^\alpha$, and vectors $\vec{X}_1, \ldots, \vec{X}_m \in \lat_0$ with $m := 2^{\ell + \alpha + 1}$; and outputs $\vec{Y}_1,\ldots, \vec{Y}_M \in \lat_\ell$ with $M := 2^\alpha$ with the following properties.
    If the input vectors $\vec{X}_1,\ldots, \vec{X}_m$ are symmetric and $0$-similar to a mixture of Gaussians over $\lat' \subseteq \lat_0$ with parameter $s > 10 \cdot 2^{(\alpha/n-1/2)\ell}\eta_\eps(\lat')$ for some (possibly unknown) sublattice $\lat' \subseteq \lat_0$ and $\eps \in (0,1/3)$; then the output distribution is $(10^\ell \eps)$-similar to a mixture of independent Gaussians over $2^{\ceil{\ell \alpha/n}}\lat' \subseteq \lat_\ell$ with parameter $2^{\ell/2}s$, and  
    \[
        \sum_{k=1}^M \expect[\|\vec{Y}_k\|^2] \leq \sum_{i = 1}^m \expect[ \|\vec{X}_i\|^2]
        \; .
    \]
\end{prop}
\begin{proof}
	The algorithm simply applies Algorithm \ref{alg:one_step} repeatedly, first using the input vectors in $\lat_0$ to obtain vectors in $\lat_1$, then using these to obtain vectors in $\lat_2$, etc., until eventually it obtains vectors $\vec{Y}_1,\ldots, \vec{Y}_M \in \lat_\ell$. The running time is clearly $m \cdot \poly(n,\ell, \log m)$, as claimed.
	
	By Lemma~\ref{lem:symmetric} and a simple induction argument, we see that every call to Algorithm~\ref{alg:one_step} results in a symmetric distribution, and the sum of the expected squared norms is non-increasing after each step. In particular,
	\[
	    \sum_{k=1}^M \expect[\|\vec{Y}_k\|^2] \leq \sum_{i = 1}^m \expect[ \|\vec{X}_i\|^2]
	    \; ,
	\]
	as needed.
	
	We suppose for induction that the distribution of the output of the $i$th call to Algorithm~\ref{alg:one_step} is $10^i \eps$-similar to a mixture of independent Gaussians over $2^{\ceil{i\alpha /n}} \lat' \subseteq 2^{\ceil{i\alpha/n}} \lat_0 \subseteq \lat_i$ with parameter $2^{i/2} s$ (which is true by assumption for $i = 0$). Then, this distribution is also $10^i \eps$-similar to a mixture of independent Gaussians over $2^{\ceil{(i+1)\alpha/n}} \lat' \subseteq 2^{\ceil{i\alpha/n}} \lat'$ (since a mixture of Gaussians over a lattice is also a mixture of Gaussians over any sublattice). Furthermore, $\eta_\eps(2^{\ceil{(i+1)\alpha/n}} \lat') = 2^{\ceil{(i+1)\alpha/n}} \eta_\eps(\lat') < 2^{i/2}s/\sqrt{2}$. Therefore, we may apply Lemma~\ref{lem:mixture_to_mixture} to conclude that the distribution of the output of the $(i+1)$st call to Algorithm~\ref{alg:one_step} is $10^{i+1} \eps$-similar to a mixture of independent Gaussians over $2^{\ceil{(i+1)\alpha/n}} \lat' \subseteq \lat_{i+1}$ with parameter $2^{(i+1)/2}s$. In particular, the final output vectors are $10^\ell \eps$-similar to a mixture of independent Gaussians over $2^{\ceil{\ell\alpha/n}}\lat'$, as needed.
\end{proof}

\subsection{The algorithm}

\begin{theorem}
    \label{thm:main_GSVP}
    For any $\eps = \eps(n) \in (0,n^{-200})$, there is a $2^{n/2+ O(n \log(n)/\log(1/\eps)) + o(n)}$-time algorithm that solves $(100 \sqrt{n} \bareta_\eps)$-GSVP. In particular, if $\eps = n^{-\omega(1)}$, then the running time is $2^{n/2 + o(n)}$.
\end{theorem}
\begin{proof}
     The algorithm takes as input a (basis for a) lattice $\lat \subset \R^n$ with $n \geq 50$ and behaves as follows. Without loss of generality, we may assume that $\eps > 2^{-n}$ and that the algorithm has access to a parameter $s > 0$ with $50\bareta_\eps(\lat) \leq s \leq 100\bareta_\eps(\lat)$. Let $\ell := \floor{\log(1/\eps)/\log(10)}-1$ and $\alpha := \ceil{n/2 + 100 n \log n/\log(1/\eps)}$.
     
     The algorithm first runs the procedure from Theorem~\ref{thm:tower_of_lattices} on input $\ell$, $\alpha$, and $\lat$, receiving as output a tower of lattices $(\lat_0, \ldots, \lat_\ell)$ with $\lat_\ell = \lat$.
    The algorithm then runs the procedure from Proposition~\ref{prop:start} on input $\lat_0$, $r := n/5$, $m := 2^{\ell + \alpha + 1}$, and parameter $s' := 2^{-\ell/2}  s$, receiving as output a sublattice $\widehat{\lat} \subseteq \lat_0$, and vectors $\vec{X}_1,\ldots, \vec{X}_m \in \widehat{\lat} \subseteq \lat_0$. 
    Finally, the algorithm runs the procedure from Proposition~\ref{prop:multi_steps} on input $(\lat_0, \ldots, \lat_\ell)$ and $\vec{X}_1,\ldots, \vec{X}_m$, receiving as output $\vec{Y}_1,\ldots, \vec{Y}_M \in \lat_\ell = \lat$. It then simply outputs the shortest non-zero vector amongst the $\vec{Y}_i \in \lat$. (If all of the $\vec{Y}_i$ are zero, the algorithm fails.)
    
    The running time of the algorithm is clearly $(m + 2^{r + o(r)}) \cdot \poly(n,\ell,\log m) = 2^{n/2 + O(n\log n /\log(1/\eps)) + o(n)}$. 
    We first show that the promise $s' \geq r^{n/r} \sqrt{n \log n} \cdot \bareta_\eps(\lat_0)$ needed to apply Proposition~\ref{prop:start} is satisfied. Indeed, by the definition of a tower of lattices, we have $\lat \subseteq 2^{\floor{\ell \alpha/n}}\lat_0$, so that
    \[
        s' \geq 50 \cdot 2^{-\ell/2}  \cdot \bareta_\eps(\lat) \geq 50 \cdot 2^{\floor{\ell\alpha/n}-\ell/2} \cdot \bareta_\eps(\lat_0) \geq r^{n/r} \sqrt{n \log n} \cdot \bareta_\eps(\lat_0)
        \; ,
    \]
    as needed.
    Therefore, the procedure from Proposition~\ref{prop:start} succeeds, i.e. we have $\bareta_\eps(\widehat{\lat}) = \bareta_\eps(\lat_0)$ and that the $\vec{X}_i$ are distributed as independent samples from $D_{\widehat{\lat},s'}$. 
    
    In particular, let $\lat' \subseteq \widehat{\lat} \subseteq \lat_0$ such that $\eta_\eps(\lat') = \bareta_\eps(\widehat{\lat}) = \bareta_\eps(\lat_0)$. Then, the distribution of $\vec{X}_1,\ldots, \vec{X}_m$ is symmetric and $0$-similar to a mixture of Gaussians over $\lat'$ with parameter $s' > 10 \cdot 2^{(\alpha/n - 1/2)\ell} \eta_\eps(\lat')$. We may therefore apply Proposition~\ref{prop:multi_steps} and see that the $\vec{Y}_1,\ldots, \vec{Y}_M \in \lat$ are $\delta$-similar to a mixture of independent Gaussians over $2^{\ceil{\ell \alpha/n}}\lat'$ with parameter $s$ and $\delta := 10^\ell \eps \leq 1/10$. Furthermore,
    \[
        \sum_{k=1}^M \expect[\|\vec{Y}_k\|^2] \leq \sum_{i=1}^m \expect[\|\vec{X}_i\|^2] \leq \frac{n m (s')^2}{2\pi} = 2^{-\ell} \cdot \frac{n m s^2}{2\pi}
        \; ,
    \]
    where the last inequality is Claim~\ref{clm:DGS_variance}.
    
    Finally, we notice that 
    \[
        s \geq 50 \bareta_\eps(\lat) \geq 50 \cdot 2^{\floor{\ell \alpha/n}} \bareta_\eps(\lat_0) = 50 \eta_\eps(2^{\floor{\ell \alpha/n}} \lat') \geq 25  \eta_\eps(2^{\ceil{\ell \alpha/n}}\lat') \geq 10 e^{\delta} \eta_{1/2}((2^{\ceil{\ell \alpha/n}}\lat')
        \; .
    \]
    Therefore, we may apply Corollary~\ref{co:short_vector} to $\vec{Y}_1,\ldots, \vec{Y}_M$ to conclude that with probability at least $1/10$, there exists $k \in \{1,\ldots, M\}$ such that
    \[
        0 < \|\vec{Y}_k\|^2 < \frac{4}{M}  \cdot \sum_{i=1}^M \expect[\|\vec{Y}_i\|^2] \leq 2^{-\ell} \cdot \frac{n m s^2}{2\pi M} \leq ns^2 \leq 100^2 n \bareta_\eps(\lat)^2
        \; .
    \]
    In other words, $\vec{Y}_k \in \lat$ is a valid solution to $(100 \sqrt{n} \bareta_\eps)$-GSVP, as needed.
\end{proof}

\begin{cor}
    There is a $2^{n/2+ o(n)}$-time algorithm that solves $\gamma$-SVP for any $\gamma = \gamma(n) > \omega(\sqrt{n \log n})$.
\end{cor}
\begin{proof}
    Theorem~\ref{thm:main_GSVP} gives an algorithm with the desired running time that finds a non-zero lattice vector with norm bounded by $ 100\sqrt{n} \bareta_\eps(\lat)$ for 
    \[
        \eps := 2^{-\gamma^2/(100^2n)} < n^{-\omega(1)}
        \; .
    \]
    The result follows from Lemma~\ref{lem:bar_eta}, which in particular tells us that 
    \[
        \bareta_\eps(\lat) \leq \sqrt{\log(1/\eps)} \lambda_1(\lat) \leq \gamma/(100\sqrt{n}) \cdot \lambda_1(\lat)
        \; ,
    \]
    as needed.
\end{proof}

\begin{cor}
\label{cor:hermite}
    There is a $2^{n/2+ o(n)}$-time algorithm that solves $\gamma$-HSVP for any $\gamma = \gamma(n) > \omega(\sqrt{n \log^3 n})$.
\end{cor}
\begin{proof}
    Theorem~\ref{thm:main_GSVP} gives an algorithm with the desired running time that finds a non-zero lattice vector with norm bounded by $100 \sqrt{n} \bareta_\eps(\lat)$ for 
    \[
        \eps := 2^{-\gamma^2/(10^{10} n \log^2 n)} < n^{-\omega(1)}
        \; .
    \]
    The result follows from Lemma~\ref{lem:bar_eta}, which in particular tells us that 
    \[
        \bareta_\eps(\lat) \leq 10\sqrt{\log(1/\eps)} (\log n + 2) \det(\lat)^{1/n} \leq \gamma/(100\sqrt{n}) \cdot \det(\lat)^{1/n}
        \; ,
    \]
    as needed (where we have assumed without loss of generality that $n$ is sufficiently large).
\end{proof}

\section{Approximate SVP via Basis Reduction}
Basis reduction algorithms solve $\delta$-(H)SVP in dimension $n$ by making polynomially many calls to a $\delta'$-SVP algorithm on lattices in dimension $k < n$. 
We will show in this section how to modify the basis reduction algorithm from~\cite{GNFindingShort08,ALNSSlideReduction19} to prove Theorem~\ref{thm:basis_reduction_intro}.

\subsection{Slide-reduced bases}

Here, we introduce our notion of a reduced basis. This differs from prior work in that we consider the possibility that the length $\ell$ of the last block is not equal to $k$, and we use HSVP reduction where other works use SVP reduction. E.g., taking $\ell=k$ and replacing (D)HSVP reduction with (D)SVP reduction in Item~\ref{item:primal_HSVP} recovers the definition from~\cite{ALNSSlideReduction19}. (Taking $\ell = k$ and $q = 0$ and replacing all (D)HSVP reduction with (D)SVP reduction recovers the original definition in~\cite{GNFindingShort08}.)

\begin{definition}[Slide reduction]
	Let $n, k, p, q, \ell$ be integers such that $n = pk + q + \ell$ with $p \geq 1, k, \ell \geq 2$ and $0\leq q \leq k-1$. Let $\delta_H \geq 1$ and $\delta_S \geq 1$. A basis $\vec{B}\in \R^{m\times n}$ is $(\delta_H, k, \delta_S, \ell)$-slide-reduced if it is size-reduced and satisfies the following four sets of constraints.
	\begin{enumerate}
		\item The block $\vec{B}_{[1,k+q+1]}$ is $\eta$-twin-reduced for $\eta:= \delta_H^{\frac{k+q-1}{k-1}}$.\label{item:first_block}
		\item For all $i \in [1,p-1]$, the block $\vec{B}_{[ik+q+1,(i+1)k+q+1]}$ is $\delta_H$-twin-reduced. \label{item:primal_HSVP}
		\item The block $\vec{B}_{[pk + q + 1, n]}$ is $\delta_S$-SVP-reduced.\label{item:last_block}
	\end{enumerate}
\end{definition}
\begin{theorem}\label{thm:SVP_slide} For any $\delta_H, \delta_S \geq 1, k \geq 2, \ell \geq 2$, if $\vec{B}\in \R^{n\times n}$ is a $(\delta_H, k, \delta_S, \ell)$-slide-reduced basis of a lattice $\lat$ with $\lambda_1(\lat(\vec{B}_{[1,n-\ell]})) > \lambda_1(\lat)$ then
		\[ \| \vec{b}_1 \| \leq \delta_S (\delta_H^2)^{\frac{n-\ell}{k-1}} \lambda_1(\lat) \; . \]
\end{theorem}
\begin{proof}
	By Fact~\ref{fact:twinsies}, $\|\vec{b}_1 \|\leq \eta^{\frac{2(k+q)}{k+q-1}} \|\vec{b}^*_{k+q+1}\| = \delta_H^{\frac{2(k+q)}{k-1}}\|\vec{b}^*_{k+q+1}\|$. Also, for all $i \in [1,p-1]$, $\|\vec{b}^*_{ik+q+1}\| \leq \delta_H^{\frac{2k}{k-1}}\|\vec{b}^*_{(i+1)k+q+1}\|$. All together we have:
	\begin{equation*}
	\|\vec{b}_1\| \le (\delta_H^2)^{\frac{k + q + (p - 1)k}{k-1}}\|\vec{b}^*_{pk+q+1}\| =  (\delta_H^2)^{\frac{n - \ell}{k-1}}\|\vec{b}^*_{pk+q+1}\|
	\end{equation*}
	Lastly, since $\lambda_1(\lat(\vec{B}_{[1,n-\ell]})) > \lambda_1(\lat)$, $\|\vec{b}^*_{pk+q+1}\| \le \delta_S \lambda_1(\lat(\basis_{[pk+q+1, n]})) \le \delta_S \lambda_1(\lat)$. The result does follow.
\end{proof}
\subsection{The slide reduction algorithm }

\label{subsec:Algorithm for GSR}
We show our algorithm for generating a slide-reduced basis. We stress that this is essentially the same algorithm as in \cite{ALNSSlideReduction19} (which itself is a generalization of the algorithm in~\cite{GNFindingShort08}) with the a slight modification that allows the last block to have arbitrary length $\ell$. Our proof for bounding the running time of the algorithm is therefore essentially identical to the proof in~\cite{GNFindingShort08,ALNSSlideReduction19}.

\begin{algorithm}
	\small
	\caption{Our slide-reduction algorithm %
		\label{alg:GSR}}
	\begin{algorithmic}[1]
		\REQUIRE Block size $k\ge 2$, slack $\eps > 0$, approximation factor $\delta_H, \delta_S \geq 1$,
		basis $\basis = (\mathbf{b}_{1}, \ldots, \mathbf{b}_{n}) \in
		\Z^{m \times n}$ of a lattice $\lat$ of rank $n =pk+q + \ell$ for $0\le q \le k-1$, and access to a $\delta_H$-HSVP oracle for lattices with rank $k$ as well as a $\delta_S$-SVP oracle for lattices with rank $\ell$.
		
		\ENSURE  A $((1 + \varepsilon)\delta_H,k, \delta_S, \ell)$-slide-reduced basis of $\lat(\basis)$.
		
		\WHILE{$\vol(\basis_{[1,ik+q]})^{2}$ is modified by the loop for some 
			$i \in [1, p]$}\label{GSR:while}

		\STATE $(1+\eps)\eta$-HSVP-reduce $\basis_{[1,k+q]}$ %
		using Alg.~\ref{alg:SDBKZ} for $\eta := (\delta_H)^{\frac{k+q-1}{k-1}}$.\label{GSR:Mordell}

		\FOR{$i = 1$ \TO $p-1$} 
		\STATE $\delta_H$-HSVP-reduce %
		$\basis_{[ik+q+1,(i+1)k+q]}$. %
		\label{GSR:HSVP}
		
		\ENDFOR
		\STATE $\delta_S$-SVP-reduce $\basis_{[pk+q+1, n]}$. \label{GSR:SVP}
		
		\IF{$\basis_{[2,k+q+1]}$ is not $(1+\eps)\eta$-DHSVP-reduced} 
		
		\STATE $(1+\eps)^{1/2} \eta$-DHSVP-reduce $\basis_{[2,k+q+1]}$ using 
		Alg.~\ref{alg:SDBKZ}. \label{step:dual_mordell}
		
		\ENDIF
		
		\FOR{$i=1$ \TO $p-1$}

		\STATE Find a new basis $\vec{C} := (\vec{b}_1,\ldots, \vec{b}_{ik+q+1},\vec{c}_{ik+q+2}, \ldots, \vec{c}_{(i+1)k+q+1},\vec{b}_{ik+q+2},\ldots,\vec{b}_n)$ of $\lat$ by $\delta_H$-DHSVP-reducing $\basis_{[ik+q+2,(i+1)k+q+1]}$.
		
		\IF{$(1+\eps) \|\vec{b}_{(i+1)k+q+1}^*\|< \| \vec{c}_{(i+1)k+q+1}^*\|$}
		
		\STATE $\basis \leftarrow \vec{C}$. \label{step:GSR_DSVP}
		
		\ENDIF
		
		\ENDFOR 
		
		\ENDWHILE\label{GSR:endwhile}
		
		\RETURN $\basis$.
	\end{algorithmic}
\end{algorithm}

\begin{theorem} \label{th:correctness of G-slide-reduction}
	For $\eps \in [1/\poly(n),1]$, Algorithm~\ref{alg:GSR} runs in polynomial time (excluding oracle calls), makes polynomially many calls to its $\delta_H$-HSVP oracle and $\delta_S$-SVP oracle, and outputs a $((1+\eps)\delta_H, k, \delta_S, \ell)$-slide-reduced basis of the input lattice $\lat$. 
\end{theorem}
\begin{proof}
	First, notice that if Algorithm~\ref{alg:GSR} ever terminates, the output must be $((1+\eps)\delta_H, k, \delta_S, \ell)$-slide-reduced basis. It remains to show that the algorithm terminates in polynomially many steps (excluding oracle calls).

	Let $\basis_{0} \in \Z^{m \times n}$ be the input basis and let $\basis \in \Z^{m \times n}$
	denote the current basis during the execution of Algorithm~\ref{alg:GSR}. 
	Following the analysis of basis reduction algorithms in \cite{LLLFactoringPolynomials82,GNFindingShort08,LNApproximatingDensest14,ALNSSlideReduction19}, 
	we consider an integral potential of the form
	\begin{equation*}
	P(\basis) := \prod_{i=1}^{p} \vol(\basis_{[1,ik+q]})^{2}  \in \mZ^{+}.\label{eq:potential for GSR}
	\end{equation*}
	At the beginning of the algorithm, the potential satisfies $\log P(\basis_{0}) \leq 2n^{2}
	\cdot \log \|\basis_{0}\|$. For each of the primal steps (i.e., Steps~\ref{GSR:Mordell},~\ref{GSR:HSVP} and~\ref{GSR:SVP}), the lattice $\lat(\basis_{[1,ik+q]})$ for any $i\ge 1$ is unchanged. Hence $P(\basis)$ does not change. On the other hand, the dual steps (i.e., Steps~\ref{step:dual_mordell} and~\ref{step:GSR_DSVP}) either leave $\vol(\basis_{[1,i k+q]})$ unchanged for all $i$ or decrease $P(\basis)$ by a multiplicative factor of at least $(1+\eps)$.
	
	Therefore, there are at most $\log P(\basis_{0})/\log
	(1+\varepsilon)$ updates on $P(\basis)$ by Algorithm~\ref{alg:GSR}. This directly implies that the algorithm makes at most $4pn^2  \log \|\basis_0\|/\log(1+\eps)$ calls to the HSVP oracle, the SVP oracle, and Algorithm~\ref{alg:SDBKZ}. 
	We then conclude that Algorithm~\ref{alg:GSR}'s running time is bounded by some polynomial in the size of input (excluding the running time of oracle calls).
	\end{proof}
\begin{cor}
	For any constant $c\geq 1$, there is a randomized algorithm that solves $\widetilde{O}(n^c)$-SVP that runs in $2^{k/2 + o(k)}$ time for $k :=  \frac{n}{c+5/(8.02)} $.
\end{cor}
\begin{proof}
	Let $\ell = \frac{0.5k}{0.802}$ and run Algorithm~\ref{alg:GSR}, instantiating the oracles with the $O(\polylog(n)\sqrt{n})$-HSVP algorithm from Corollary~\ref{cor:hermite} and the $O(1)$-SVP algorithm from~\cite{LWXZShortestLattice11} to get a $((1+\eps)\polylog(k)\sqrt{k}, k, O(1), \ell)$-slide-reduced basis $\basis$ for any input lattice $\lat$. Now consider two cases:
	\begin{itemize}
		\item $\lambda_1(\lat(\vec{B}_{[1,n-\ell]})) > \lambda_1(\lat)$: By Theorem~\ref{thm:SVP_slide}, $\| \vec{b}_1 \| \leq \delta_S (\delta_H^2)^{\frac{n-\ell}{k-1}} \lambda_1(\lat) \leq O(\polylog(k)^cn^c) \lambda_{1}(\lat)$ as desired.
		\item $\lambda_1(\lat(\vec{B}_{[1,n-\ell]})) = \lambda_1(\lat)$: Then we repeat the algorithm on the lattice $\lat(\vec{B}_{[1,n-\ell]})$ with lower dimension. This can happen at most $n/\ell$ times, introducing at most a polynomial factor in the running time.
	\end{itemize}
	For the running time, the algorithm from Corollary~\ref{cor:hermite} runs in time $2^{0.5 k + o(k)}$. The algorithm from \cite{LWXZShortestLattice11} runs in time $2^{0.802\ell+o(\ell)}$, which is the same as $2^{0.5 k + o(k)}$, by our choice of $\ell$. This completes the proof. 
\end{proof}

\newcommand{\etalchar}[1]{$^{#1}$}


\begin{thebibliography}{GMPW20}
	
	\bibitem[ADRS15]{ADRSSolvingShortest15}
	Divesh Aggarwal, Daniel Dadush, Oded Regev, and Noah {Stephens-Davidowitz}.
	\newblock Solving the {Shortest Vector Problem} in $2^n$ time via {Discrete
		Gaussian Sampling}.
	\newblock In {\em {{STOC}}}, 2015.
	\newblock \url{http://arxiv.org/abs/1412.7994}.
	
	\bibitem[AKS01]{AKSSieveAlgorithm01}
	Mikl{\'o}s Ajtai, Ravi Kumar, and D.~Sivakumar.
	\newblock A sieve algorithm for the {{Shortest Lattice Vector Problem}}.
	\newblock In {\em {{STOC}}}, 2001.
	
	\bibitem[ALNS20]{ALNSSlideReduction19}
	Divesh Aggarwal, Jianwei Li, Phong~Q. Nguyen, and Noah {Stephens-Davidowitz}.
	\newblock Slide reduction, revisited--{F}illing the gaps in {SVP}
	approximation.
	\newblock In {\em {{CRYPTO}}}, 2020.
	
	\bibitem[AS04]{AS04}
	Noga Alon and Joel~H Spencer.
	\newblock {\em The probabilistic method}.
	\newblock John Wiley \& Sons, 2004.
	
	\bibitem[AS18]{ASJustTake18}
	Divesh Aggarwal and Noah {Stephens-Davidowitz}.
	\newblock Just take the average! {A}n embarrassingly simple $2^n$-time
	algorithm for {SVP} (and {CVP}).
	\newblock In {\em {{SOSA}}}, 2018.
	\newblock \url{http://arxiv.org/abs/1709.01535}.
	
	\bibitem[AUV19]{AUV19}
	Divesh Aggarwal, Bogdan Ursu, and Serge Vaudenay.
	\newblock Faster sieving algorithm for approximate {{SVP}} with constant
	approximation factors.
	\newblock \url{https://eprint.iacr.org/2019/1028}, 2019.
	
	\bibitem[BLP{\etalchar{+}}13]{BLPRS13}
	Zvika Brakerski, Adeline Langlois, Chris Peikert, Oded Regev, and Damien
	Stehl{\'e}.
	\newblock Classical hardness of {Learning with Errors}.
	\newblock In {\em STOC}, 2013.
	
	\bibitem[DR16]{DRStrongReverse16}
	Daniel Dadush and Oded Regev.
	\newblock Towards strong reverse {Minkowski}-type inequalities for lattices.
	\newblock In {\em {{FOCS}}}, 2016.
	\newblock \url{http://arxiv.org/abs/1606.06913}.
	
	\bibitem[DRS14]{DRS14}
	Daniel Dadush, Oded Regev, and Noah Stephens{-}Davidowitz.
	\newblock On the {Closest Vector Problem} with a distance guarantee.
	\newblock In {\em CCC}, 2014.
	
	\bibitem[GMPW20]{GMPWImprovedDiscrete20}
	Nicholas Genise, Daniele Micciancio, Chris Peikert, and Michael Walter.
	\newblock Improved discrete {{Gaussian}} and subgaussian analysis for lattice
	cryptography.
	\newblock In {\em {{PKC}}}, 2020.
	\newblock \url{https://eprint.iacr.org/2020/337}.
	
	\bibitem[GN08]{GNFindingShort08}
	Nicolas Gama and Phong~Q. Nguyen.
	\newblock Finding short lattice vectors within {Mordell's} inequality.
	\newblock In {\em {{STOC}}}, 2008.
	
	\bibitem[GPV08]{GPVTrapdoorsHard08}
	Craig Gentry, Chris Peikert, and Vinod Vaikuntanathan.
	\newblock Trapdoors for hard lattices and new cryptographic constructions.
	\newblock In {\em {{STOC}}}, 2008.
	\newblock \url{https://eprint.iacr.org/2007/432}.
	
	\bibitem[Kan83]{KanImprovedAlgorithms83}
	Ravi Kannan.
	\newblock Improved algorithms for integer programming and related lattice
	problems.
	\newblock In {\em {{STOC}}}, 1983.
	
	\bibitem[LLL82]{LLLFactoringPolynomials82}
	Arjen~K. Lenstra, Hendrik~W. Lenstra, Jr., and L{\'a}szl{\'o} Lov{\'a}sz.
	\newblock Factoring polynomials with rational coefficients.
	\newblock {\em Mathematische Annalen}, 261(4), 1982.
	
	\bibitem[LN14]{LNApproximatingDensest14}
	Jianwei Li and Phong~Q. Nguyen.
	\newblock Approximating the densest sublattice from {Rankin’s} inequality.
	\newblock {\em LMS J. of Computation and Mathematics}, 17(A), 2014.
	
	\bibitem[Lov86]{LovAlgorithmicTheory86}
	L{\'a}szl{\'o} Lov{\'a}sz.
	\newblock {\em An algorithmic theory of numbers, graphs and convexity}.
	\newblock Society for Industrial and Applied Mathematics, 1986.
	
	\bibitem[LWXZ11]{LWXZShortestLattice11}
	Mingjie Liu, Xiaoyun Wang, Guangwu Xu, and Xuexin Zheng.
	\newblock Shortest lattice vectors in the presence of gaps.
	\newblock \url{http://eprint.iacr.org/2011/139}, 2011.
	
	\bibitem[MP13]{MPHardnessSIS13}
	Daniele Micciancio and Chris Peikert.
	\newblock Hardness of {SIS} and {LWE} with small parameters.
	\newblock In {\em {{CRYPTO}}}, 2013.
	
	\bibitem[MR07]{MR07}
	Daniele Micciancio and Oded Regev.
	\newblock Worst-case to average-case reductions based on {G}aussian measures.
	\newblock {\em SIAM Journal on Computing}, 37(1):267--302, 2007.
	
	\bibitem[MV13]{MVDeterministicSingle13}
	Daniele Micciancio and Panagiotis Voulgaris.
	\newblock A deterministic single exponential time algorithm for most lattice
	problems based on {Voronoi} cell computations.
	\newblock {\em SIAM J. on Computing}, 42(3), 2013.
	
	\bibitem[MW16]{MWPracticalPredictable16}
	Daniele Micciancio and Michael Walter.
	\newblock Practical, predictable lattice basis reduction.
	\newblock In {\em Eurocrypt}, 2016.
	\newblock \url{http://eprint.iacr.org/2015/1123}.
	
	\bibitem[NIS18]{NISPostQuantumCryptography18}
	Computer Security~Division NIST.
	\newblock Post-quantum cryptography.
	\newblock \url{https://csrc.nist.gov/Projects/Post-Quantum-Cryptography}, 2018.
	
	\bibitem[NV08]{NVSieveAlgorithms08}
	Phong~Q. Nguyen and Thomas Vidick.
	\newblock Sieve algorithms for the {Shortest Vector Problem} are practical.
	\newblock {\em J. Mathematical Cryptology}, 2(2), 2008.
	
	\bibitem[Pei16]{PeiDecadeLattice16}
	Chris Peikert.
	\newblock A decade of lattice cryptography.
	\newblock {\em Foundations and Trends in Theoretical Computer Science}, 10(4),
	2016.
	
	\bibitem[PS09]{PSSolvingShortest09}
	Xavier Pujol and Damien Stehl{\'e}.
	\newblock Solving the {{Shortest Lattice Vector Problem}} in time $2^{2.465
		n}$, 2009.
	\newblock \url{http://eprint.iacr.org/2009/605}.
	
	\bibitem[Reg09]{Reg09}
	Oded Regev.
	\newblock On lattices, learning with errors, random linear codes, and
	cryptography.
	\newblock {\em Journal of the ACM (JACM)}, 56(6):34, 2009.
	
	\bibitem[RS17]{Regev:2017:RMT:3055399.3055434}
	Oded Regev and Noah {Stephens-Davidowitz}.
	\newblock A reverse {M}inkowski theorem.
	\newblock In {\em STOC}, 2017.
	
	\bibitem[Sch87]{SchHierarchyPolynomial87}
	Claus-Peter Schnorr.
	\newblock A hierarchy of polynomial time lattice basis reduction algorithms.
	\newblock {\em Theor. Comput. Sci.}, 53(23), 1987.
	
	\bibitem[SE94]{SELatticeBasis94}
	Claus-Peter Schnorr and M.~Euchner.
	\newblock Lattice basis reduction: {{Improved}} practical algorithms and
	solving subset sum problems.
	\newblock {\em Mathmatical Programming}, 66, 1994.
	
	\bibitem[{Ste}17]{NSDthesis}
	Noah {Stephens-Davidowitz}.
	\newblock {\em On the Gaussian measure over lattices}.
	\newblock Phd thesis, New York University, 2017.
	
	\bibitem[WLW15]{WLWFindingShortest15}
	Wei Wei, Mingjie Liu, and Xiaoyun Wang.
	\newblock Finding shortest lattice vectors in the presence of gaps.
	\newblock In {\em {{CT}}-{{RSA}}}, 2015.
	
\end{thebibliography}
\end{document}